\documentclass{article}
\usepackage{arxiv}

\usepackage{graphicx}%
\usepackage{multirow}%
\usepackage{amsmath,amssymb,amsfonts}%
\usepackage{amsthm}%
\usepackage{mathrsfs}%
\usepackage[title]{appendix}%
\usepackage{xcolor}%
\usepackage{bm}
\usepackage{stmaryrd}
\usepackage{textcomp}%
\usepackage{manyfoot}%
\usepackage{mathtools}
\usepackage{booktabs}%
\usepackage{algorithm}%
\usepackage{algorithmicx}%
\usepackage{algpseudocode}%
\usepackage{listings}%
\usepackage{hyperref}       
\newtheorem{theorem}{Theorem}
\newtheorem{remark}{Remark}
\newtheorem{lemma}{Lemma}
\title{Entropy-conservative high-order methods for high-enthalpy gas flows}

\author{\href{https://orcid.org/0000-0002-1434-7166}{\includegraphics[scale=0.06]{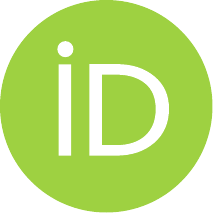}\hspace{1mm}Georgii Oblapenko}\thanks{
Preprint submitted to Computers and Fluids}
 , \href{https://orcid.org/0000-0002-1434-7166}{\includegraphics[scale=0.06]{orcid.pdf}\hspace{1mm}Manuel Torrilhon}\\
	Chair of Applied and Computational Mathematics, RWTH Aachen, \\
	Schinkelstrasse 52, 52062 Aachen, Germany \\
	$^\ast$Corresponding author. E-mail: \texttt{oblapenko@acom.rwth-aachen.de}
}

\renewcommand{\shorttitle}{Entropy-conservative high-order methods for high-enthalpy gas flows}

\hypersetup{
pdftitle={Entropy-conservative high-order methods for high-enthalpy gas flows},
pdfsubject={},
pdfauthor={Georgii Oblapenko, Manuel Torrilhon},
pdfkeywords={Gas dynamics, discontinuous Galerkin, high enthalpy, entropy conservation},
}

\begin{document}
\maketitle

\begin{abstract}
A framework for numerical evaluation of entropy-conservative volume fluxes in gas flows with internal energies is developed, for use with high-order discretization methods. The novelty of the approach lies in the ability to use arbitrary expressions for the internal degrees of freedom of the constituent gas species. The developed approach is implemented in an open-source discontinuous Galerkin code for solving hyperbolic equations. Numerical simulations are carried out for several model 2-D flows and the results are compared to those obtained with the finite volume-based solver DLR TAU.
\end{abstract}





\section{Introduction}\label{sec1}
Accurate simulation of non-equilibrium flows with chemical reactions and internal degrees of freedom is crucial for a multitude of aerospace applications~\cite{candler_computation_1991,nagnibeda2009non,karl2024sustainable}. For such flows, using higher-order discontinuous Galerkin (DG) methods is an attractive prospect, as it enables a more accurate simulation of turbulent effects, lessens the requirements on grid-shock alignment, improves computational efficiency due to data locality, and adds additional solution adaptivity capabilities via $p$-adaptivity~\cite{wang_high-order_2013,hempert2017simulation,kronbichler2021high,basile2022unstructured,mossier2022p,hoskin2023discontinuous,peyvan2023high}. DG methods are also the natural framework for polynomial expansion-based uncertainty quantification and sensitivity analysis studies~\cite{smith2013uncertainty,kusch2020filtered}, an area of research with growing importance, as larger and larger parts of the R\&D cycle in aerospace are carried out via simulation tools.
The less stringent requirements of DG methods with regards to shock-grid alignment~\cite{ching2019shock} are also useful for uncertainty quantification and sensitivity analysis studies, as those oftentimes require automated simulation of hundreds and thousands of test cases, where automated well-aligned mesh generation is possible only for simple geometries~\cite{cortesi2020forward}.

Development of DG methods suitable for aerospace applications is still an active area of research, as questions of numerical stability, shock capturing, and insufficient numerical diffusion inherent to higher-order methods need to be addressed. Recent achievements in these areas include the DG spectral element method with Legendre-Gauss-Lobatto quadratures (DGSEM-LGL)~\cite{friedrich2018entropy,gassner2021novel,renac2021entropy}, artificial viscosity approaches~\cite{michoski2016comparison,ching2019shock,bai2022continuous}, shock filtering~\cite{michoski2016comparison,bohm2019multi}, and subcell limiting techniques~\cite{hennemann2021provably,rueda2022subcell,lin2024high,rueda2024monolithic}.

Another question is the choice of the numerical flux functions: from a physics point of view, a crucial desired property of a numerical scheme is entropy conservation~\cite{tadmor2003entropy,jameson2008formulation,ismail2009affordable,fisher2013high}, especially in the context of DG methods~\cite{chandrashekar2013kinetic,chan2018discretely,ranocha2018comparison,gassner2021novel}. For summation-by-part (SBP) schemes, ensuring entropy conservation/entropy stability at the semi-discrete level requires entropy-conservative numerical flux functions for the volume and surface fluxes~\cite{ranocha2018comparison,gassner2021novel}. Augmented by a dissipative surface flux function, the spatial semi-discretization then guarantees entropy stability~\cite{chan2018discretely}.
Multiple entropy-conservative flux functions have been developed~\cite{ranocha2018comparison}, but they usually assume restrictions on the description of the internal energy of the gas flow. Quite often, a constant ratio of specific heats is assumed. Whilst this greatly simplifies the derivation of numerical fluxes, it is not applicable to the kind of high-temperature flows encountered in real-life combustion and aerospace applications, where excitation internal degrees of freedom plays a significant role~\cite{candler_computation_1991,mcbride2002nasa,nagnibeda2009non,karl2024sustainable}. Therefore, development of new numerical fluxes that possess attractive properties, such as preservation of entropy, kinetic energy, and pressure equilibrium, for gases with non-trivial internal energy functions and/or equations of state is of very high relevance to the applied CFD community in order to drive the development of new robust higher-order methods for large-scale simulations.

Several works have concerned themselves with exactly such developments. In~\cite{gouasmi2020formulation,ching2024positivity}, entropy-conservative schemes were developed for gases with internal degrees of freedom were the internal energy is modelled by a polynomial function of  temperature~\cite{mcbride2002nasa}.  In~\cite{johnson2020conservative}, the DG method was applied to the simulation of reacting mixtures with special attention paid to the numerical evaluation of fluxes that avoid generating spurious pressure oscillations; the internal energies were also modelled using NASA polynomials.
In~\cite{aiello2024entropy}, entropy-conservative discretizations have been proposed for gases with an arbitrary (non-ideal gas law) equation of state; the impact of internal degrees of freedom on the ideal gas effects were also modelled by polynomial functions of temperature.
In~\cite{osti_1763209}, entropy-stable flux functions were derived for gas mixtures with internal energies described either by polynomial fits, or with vibrational spectra modelled by the infinite harmonic oscillator model. Flux functions for mixtures with thermal non-equilibrium between the vibrational and rotational-translational modes were derived as well.
Similarly, in~\cite{peyvan2023high}, an approach has been proposed to construct entropy stable flux functions for mixtures of molecules with vibrational spectra described by infinite harmonic oscillators. In~\cite{abgrall2022reinterpretation} an alternative approach for evaluating entropy-conservative fluxes for gases with arbitrary equations of state was proposed, based on computation of correction terms as solutions of an optimization problem.

Thus, recent advancements in the field of structure-preserving methods for gas dynamics have significantly expanded the range of applicability of such entropy-conservative schemes to real gases with complex equations of state and dependencies of internal energy on the gas temperature. However, they usually still rely either on polynomial fits or the infinite harmonic oscillator model, which is not the most accurate description of real-life vibrational spectra of molecules. In addition, the approaches used to derive the flux functions does not easily generalize to other expressions for internal energy spectra.

In the present work, a new approach to derive entropy-conservative flux schemes is developed, based on interpolation of the specific heats and entropy integrals. It is applied to the simulation of inviscid high-enthalpy flows over a cylinder, and compared to results obtained with the second order finite volume-based DLR TAU solver~\cite{mack2002validation}.

The paper is structured as follows. First, the flow equations for a high-temperature gas flow are presented, the thermodynamic properties of high-temperature gases are discussed, and a brief overview of the impact of various modelling assumptions is given. Then, an entropy-conservative flux is derived for the equations in question, and its computational properties are analyzed. Next, a short overview of the used simulation framework is given and results of numerical simulations are presented. Finally, future work prospects are discussed.

\section{Flow equations and thermodynamic properties}
We consider a two-dimensional flow of a single-species inviscid gas.

The compressible Euler equations governing such a flow are given by
\begin{equation}
    \frac{\partial}{\partial t}\mathbf{u} + \frac{\partial}{\partial x}\mathbf{f}_x + \frac{\partial}{\partial y}\mathbf{f}_y = 0.\label{eqns:euler-1}
\end{equation}
Here $t \in [0,t_{\mathrm{max}}]$ is the time, $x$ and $y$ are the spatial coordinates in the flow domain $\Omega \subset \mathbb{R}^2$, $\mathbf{u}$ is the vector of conservative variables, and $\mathbf{f}_x$, $\mathbf{f}_y$ are the inviscid fluxes.

The vector of conservative variables  $\mathbf{u} \in \mathbb{R}^4$ is given by
\begin{equation}
     \mathbf{u} = \left(\rho, \rho v_x, \rho v_y, E \right)^{\mathrm{T}},
\end{equation}
where $\rho$ is the density, $v_x$ and $v_y$ are the flow velocities in the $x$ and $y$ directions, and $E=\rho e=\rho \varepsilon_{\mathrm{int}} + \rho v^2 / 2$ is the total flow energy. Here, $e$ is the specific energy, and $\varepsilon_{\mathrm{int}}$ is the specific internal energy.
The inviscid fluxes are given by
\begin{equation}
     \mathbf{f}_x = \left(\rho v_x, \rho v_x^2 + p, \rho v_x v_y, (E+p)v_x \right)^{\mathrm{T}},\label{eqns:euler-flux_x}
\end{equation}
\begin{equation}
     \mathbf{f}_y = \left(\rho v_y, \rho v_x v_y, \rho v_y^2 + p, (E+p)v_y \right)^{\mathrm{T}}.\label{eqns:euler-flux_y}
\end{equation}
Here $p$ is the pressure. We use the ideal gas law to relate pressure, density, and temperature: $p=nkT$, where $n$ is the number density ($n=\rho / m$, where $m$ is the mass of the constituent gas species), $k$ is the Boltzmann constant, and $T$ is the flow temperature. All that is left to fully close the system of equations is a relation between the specific internal energy and the temperature.

\subsection{Internal energy models}
We now consider several models for $\varepsilon_{\mathrm{int}}(T)$. The simplest case is that of the calorically perfect gas, where $\varepsilon_{\mathrm{int}}(T) = c_v T$, where $c_v$ is the constant-valued specific heat at constant volume. The ratio of specific heats $\gamma=(c_v + k/m)/c_v$ in this case is also temperature-independent. The calorically perfect gas case, especially with $c_v=5/2 k/m$ (corresponding to $\gamma=1.4$), is the standard one in a large amount of literature concerned with CFD methods in general and entropy-conservative formulations in particular, and we do not consider it in detail. The factor $5/2$ comes of the contribution of the translational degrees of freedom ($3/2$) and the fully excited rotational degrees of freedom ($2/2$) for molecular gases.

In the present work, we assume that our gas consists of diatomic molecules, which possess only rotational and vibrational degrees of freedom. That is, we neglect the impact of electronic excitation and can write $\varepsilon_{\mathrm{int}}=\varepsilon_{\mathrm{tr}} + \varepsilon_{\mathrm{rot}} + \varepsilon_{\mathrm{vibr}}$, where $\varepsilon_{\mathrm{tr}}=\frac{3}{2} kT/m$ is the specific translational energy, and $\varepsilon_{\mathrm{rot}}$ and $\varepsilon_{\mathrm{vibr}}$ are the specific rotational and vibrational energies, respectively. Moreover, we assume a ``rigid rotor'' model~\cite{nagnibeda2009non}, that is, that the rotational and vibrational degrees of freedom are uncoupled. We also assume that the rotational mode is fully excited, leading to $\varepsilon_{\mathrm{rot}}=kT/m$. As a consequence, the specific heat of the rotational degrees of freedom simplifies to $c_{v,\mathrm{rot}}=k/m$. For a more detailed discussion of the validity of these assumptions, at least with regard to viscous fluid properties, the reader is referred to~\cite{kustova2017applicability}.

For the vibrational degrees of freedom, we consider the two following models:
\begin{enumerate}
    \item Infinite harmonic oscillator: 
    \begin{equation}
        \varepsilon_{\mathrm{vibr}}=\frac{k}{m}\frac{\theta_v}{\exp\left(\frac{\theta_v}{T}\right)-1}, \quad c_{v,\mathrm{vibr}}=\frac{k}{m}\frac{\theta_v^2 \exp\left(\frac{\theta_v}{T}\right)}{T^2\left(\exp\left(\frac{\theta_v}{T}\right)-1\right)^2},\label{eq:iho}
    \end{equation}
    where $\theta_v$ is the characteristic vibrational temperature of the molecules.
    \item Cut-off oscillator:
    \begin{equation*}
        \varepsilon_{\mathrm{vibr}}=\frac{1}{mZ_{\mathrm{vibr}}}\sum_{i=0}^{i_{max}} \varepsilon_{v,i}\exp\left(-\frac{\varepsilon_{v,i}}{kT}\right), \quad Z_{\mathrm{vibr}}=\sum_{i=0}^{i_{max}}\exp\left(-\frac{\varepsilon_{v,i}}{kT}\right),\label{eq:cutoffosc}
    \end{equation*}
    \begin{equation}
        c_{v,\mathrm{vibr}}=\frac{1}{mZ_{\mathrm{vibr}}}\frac{1}{kT^2}\left(\sum_{i=0}^{i_{max}}\varepsilon_{v,i}^2\exp\left(-\frac{\varepsilon_{v,i}}{kT}\right) - \left(\sum_{i=0}^{i_{max}}\varepsilon_{v,i}\exp\left(-\frac{\varepsilon_{v,i}}{kT}\right)\right)^2 \right).
    \end{equation}
    Here $Z_{\mathrm{vibr}}$ is the vibrational partition function, the summation over $i$ denotes the summation over all the allowed discrete vibrational states, $\varepsilon_{v,i}$ is the vibrational energy of the molecule in vibrational state $i$, and $i_{max}$ is the maximum allowed vibrational state. Usually this is chosen based on the dissociation energy $D$ of the molecule, such that $\varepsilon_{v,i_{max}}<D<\varepsilon_{v,i_{max}+1}$.

    In the case of $\varepsilon_{v,i} = \left(i + \frac{1}{2}\right) k \theta_v$, where $\theta_v$ is the characteristic vibrational temperature of the molecular species, one obtains the \textit{cut-off harmonic oscillator model}. In case a non-linear dependence of the vibrational energy on the level number is assumed, one obtains a so-called \textit{cut-off anharmonic oscillator model}. In the present work, second-order terms are considered in the case of the anharmonic oscillator model: $\varepsilon_{v,i} = \left(i + \frac{1}{2}\right) k \theta_v - \left(i + \frac{1}{2}\right)^2 k \theta_{v,\mathrm{anh}}$, where $\theta_{v,\mathrm{anh}}$ is the characteristic temperature of the anharmonic correction.

\end{enumerate}

\begin{figure}[h]
    \centering
    \includegraphics[width=0.7\textwidth]{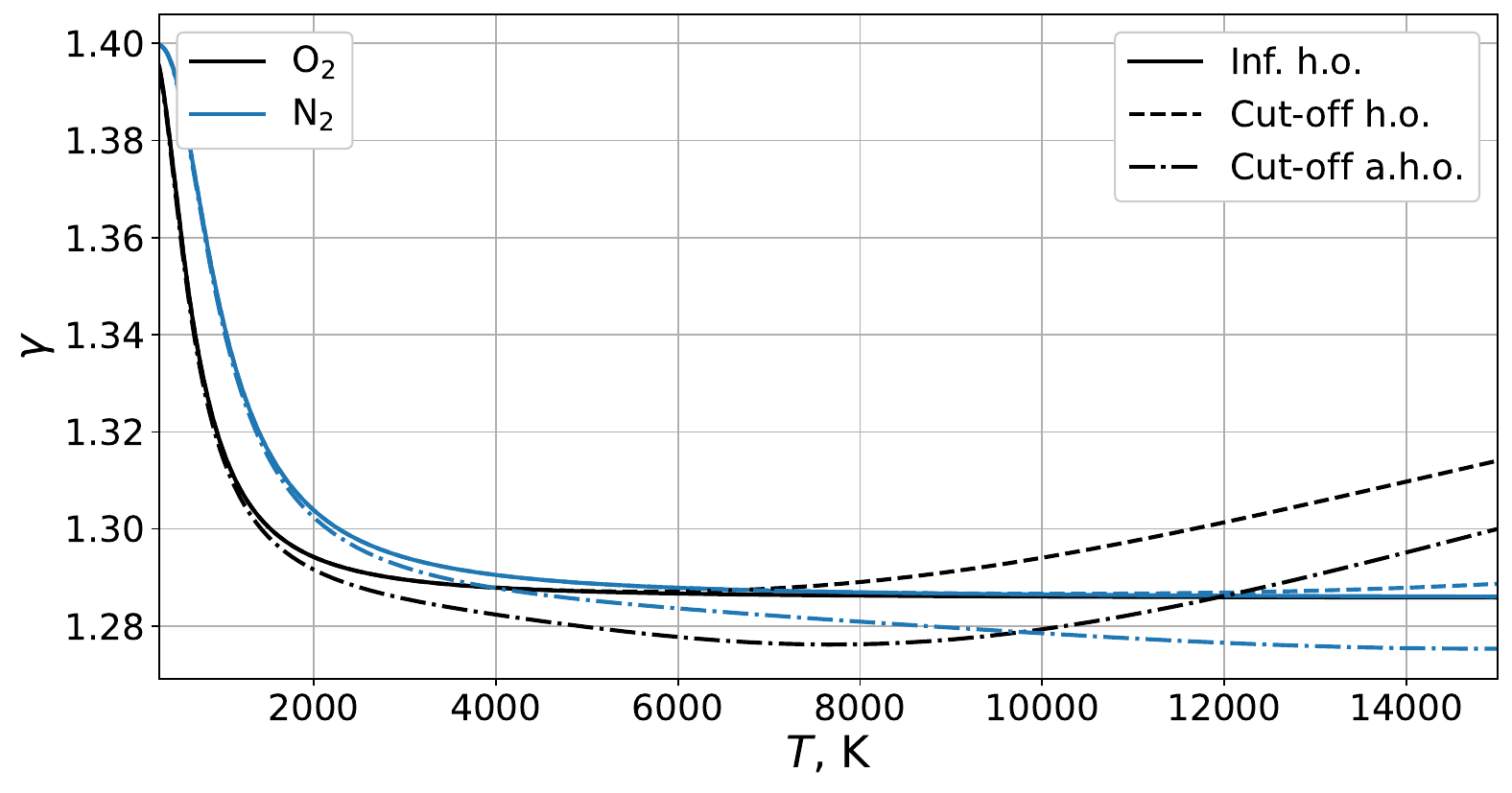}
    \caption{The heat capacity ratio $\gamma$ as a function of temperature for different models of the vibrational energy spectrum of oxygen (black lines) and nitrogen (blue lines).}
    \label{fig:gamma}
\end{figure}

In the present work we consider two molecular species: O$_2$ and N$_2$. For O$_2$, the following values were used to determine the vibrational spectrum~\cite{capitelli2005tables}: a dissociation energy $D$ of 59364~K, a characteristic vibrational temperature $\theta_v$ of 2273.5~K, and a characteristic temperature of the anharmonic correction $\theta_{v,\mathrm{anh}}$ of $17.366$~K. For N$_2$, the following values were used: $D=113252$~K, $\theta_v=3393.48$~K, $\theta_{v,\mathrm{anh}}$ of $20.603$~K. 

Figure~\ref{fig:gamma} shows the impact the choice of the vibrational energy spectrum description has on the thermodynamic properties of the gas (in this case, the heat capacity ratio $\gamma$). Three different models of increasing complexity are considered: the infinite harmonic oscillator, the cut-off harmonic oscillator, and the cut-off anharmonic oscillator. For both of the cut-off models, the cut-off was determined by the dissociation energy $D$. 
It can be seen that the heat capacity ratio strongly deviates from the room temperature value of 1.4, with all the cut-off models exhibiting non-monotone behaviour. For nitrogen, the effects of vibrational excitation and vibrational spectrum cut-off appear at higher temperatures than for oxygen, due to its higher characteristic vibrational temperature and larger dissociation energy. Due to the larger dissociation energy of nitrogen, the difference between the infinite and cut-off harmonic oscillator models is also not significant in the temperature range considered, appearing only for $T>12000$~K, whereas for oxygen the difference becomes noticeable already for temperatures above 6000~K. The impact of anharmonicity is more pronounced for both molecules.  It can thus be concluded that whilst the infinite harmonic oscillator model does allow for simulation of more complex thermodynamical phenomena compared to a calorically perfect gas model, it is noticeably less accurate than the more detailed cut-off oscillator models for high temperatures.
Note that the presented results do not account for effects such as electronic excitation, which further complicates the thermodynamic behaviour of gases at temperatures higher than 7000-10000~K~\cite{mcbride2002nasa,capitelli2005tables}.

\subsection{Implementation of internal energy models in a CFD solver}
Accounting for internal energy effects in a flow solver involves implementing several  functions: 1) one to compute the flow temperature $T$ from the conservative variables, 2) one to compute a specific internal energy given the flow temperature $T$, 3) one to compute the specific heats given the flow temperature $T$, and 4) one to compute the flow entropy given the temperature and conservative variables. 

In general, no closed-form expression exists to compute $T$ from $\varepsilon_{\mathrm{int}}$. As such, existing packages for high-temperature flow simulation, such as Mutation++~\cite{scoggins2020mutation}, use a non-linear solver to obtain the flow temperature. 
In the present work, a single-species gas is considered, and an approach based on linear interpolation between tabulated values is used due to the lower computational cost.

First, reasonable minimum and maximum values of $T$ are chosen, e.g., $T_{\mathrm{min}}=10$~K, $T_{\mathrm{max}}=50000$~K, as well as a temperature discretization step $\Delta T$. The full specific internal energy $\varepsilon_{\mathrm{int}}=\varepsilon_{\mathrm{int}}(T)$ is computed over the specified temperature range with the corresponding step in $\Delta T$ and stored, as well as the total specific heat $c_{v}=c_{v}(T)$. This allows for efficient computation of temperature, a primitive variable, from energy, a conservative variable, as the Newton solver used to solve for temperature can use the linear interpolation for $c_v$ and $\varepsilon_{\mathrm{int}}$ instead of the potentially computationally more expensive analytical expressions.

\section{Entropy-conservative flux}
Having defined our system of flow equations and its closures, we are interested in developing an entropy-conservative flux function. In this, we follow the methodology of Tadmor~\cite{tadmor2003entropy}, with the help of the more generic formulation of the entropy as given in~\cite{peyvan2023high}.

\subsection{Entropy and entropy variables}
The physical entropy $s$ is defined as the specific quantity
\begin{equation}
    s = \int_0^T \frac{c_v(\tau)}{\tau}\mathrm{d}\tau - \frac{k}{m} \ln \rho = \eta(T) -  \frac{k}{m} \ln \rho. \label{eq:entropy-def}
\end{equation}
We will hereafter refer to the $\int_0^T \frac{c_v(\tau)}{\tau}\mathrm{d}\tau $ as the ``integral part'' of the entropy. As it will play a central role in the following derivations, we denote it by $\eta(T)$ to simplify the notation. We also define the mathematical entropy as the volume density $\mathfrak{s}=-\rho s$. 

The specific entropy satisfies the entropy balance equation~\cite{harten1998convex}
\begin{equation}
    \frac{\partial s}{\partial t} + \mathbf{v} \cdot \nabla s = \sigma \geq 0,
\end{equation}
where $\mathbf{v}=(v_x, v_x)$ is the velocity vector, and $\sigma$ is the entropy production. The entropy production $\sigma$ depends on the physics of the problem being considered, i.e. presence of shockwaves, chemical reactions, etc., and we thus leave it unspecified.

In order to derive an entropy-conservative flux, we first need to compute entropy variables, defined by
\begin{equation}
    \bm{\omega}=\frac{\partial \mathfrak{s}}{\partial \mathbf{u}}.
\end{equation}
As before, our vector of conservative variables is the following
\begin{equation}
    \mathbf{u} = \left(\rho, \rho v_x, \rho v_y, E\right)^{\mathrm{T}} = \left(\rho, \rho v_x, \rho v_y, \rho \varepsilon_{\mathrm{int}} + \frac{1}{2}\rho v^2 \right)^{\mathrm{T}},
\end{equation}
and we define our vector of primitive variables as
\begin{equation}
    \mathbf{z} = \left(\rho, v_x, v_y, T\right)^{\mathrm{T}}.
\end{equation}
Using the chain rule, we can write the entropy variable vector~\cite{harten1998convex} 
\begin{equation}
    \bm{\omega}=\frac{\partial \mathfrak{s}}{\partial \mathbf{u}}=\frac{\partial \mathfrak{s}}{\partial \mathbf{z}}\left(\frac{\partial \mathbf{u}}{\partial \mathbf{z}}\right)^{-1}.\label{eq:omega-chainrule}
\end{equation}
Computing the derivatives of the conservative variables w.r.t. the primitive variables is straightforward
\begin{equation}
    \frac{\partial \mathbf{u}}{\partial \mathbf{z}} = 
    \begin{pmatrix}
    1 & 0 & 0 & 0\\
    v_x & \rho & 0 & 0\\
    v_y & 0 & \rho & 0\\
    \varepsilon_{\mathrm{int}} + \frac{1}{2} v^2 & \rho v_x & \rho v_y & \rho c_v
    \end{pmatrix}
\end{equation}
Inverting, we obtain
\begin{equation}
    \left(\frac{\partial \mathbf{u}}{\partial \mathbf{z}}\right)^{-1} = 
    \begin{pmatrix}
    1 & 0 & 0 & 0\\
    -\frac{v_x}{\rho} & \frac{1}{\rho} & 0 & 0\\
    -\frac{v_y}{\rho} & 0 & \frac{1}{\rho} & 0\\
    \frac{\frac{1}{2}v^2 - \varepsilon_{\mathrm{int}}}{\rho c_v} & -\frac{v_x}{\rho c_v} & -\frac{v_y}{\rho c_v} & \frac{1}{\rho c_v}
    \end{pmatrix}
\end{equation}
Making use of the Leibniz integral rule, we get
\begin{equation}
    \frac{\partial \mathfrak{s}}{\partial \mathbf{z}} = \left( - s + \frac{k}{m}, 0, 0, -\rho\frac{c_v}{T} \right).
\end{equation}
Applying (\ref{eq:omega-chainrule}) and omitting any constant terms independent of the flow variables, we obtain the vector of entropy variables.
\begin{equation}
    \bm{\omega} =  \left(\omega_1, \omega_2, \omega_3, \omega_4 \right) = \left(-s(T) + \frac{\varepsilon_{\mathrm{int}}(T)-\frac{1}{2}v^2}{T}, \frac{v_x}{T}, \frac{v_y}{T}, -\frac{1}{T} \right)^{\mathrm{T}}. \label{eq:omega-vector}
\end{equation}

\subsection{Derivation of entropy-conservative flux}
The equation for entropy-conserving flux $\mathbf{F}^{\mathrm{num},EC,j} \in \mathbb{R}^4$ in spatial direction $j$ ($j=x,y$) is given by the following compatibility condition~\cite{tadmor2003entropy}:
\begin{equation}
    \llbracket \bm{\omega}^{\mathrm{T}} \rrbracket \cdot \mathbf{F}^{\mathrm{num},EC,j} - \llbracket \psi_j \rrbracket = 0,\label{eqn:ec-fluxes}
\end{equation}
where $\llbracket a \rrbracket$ = $a_+ - a_-$ is the jump in the value of a flow variable $a$, and $a_-$ and $a_+$ denote the values of a flow variable $a$ at two points between which the flux is computed, that is, the flux $\mathbf{F}^{\mathrm{num},EC,j}$ depends on both flow states $\mathbf{u}_{\pm}$.  The flux potentials are given by~\cite{peyvan2023high}:
\begin{equation}
    \psi_j = \frac{k}{m}\rho v_j.
\end{equation}
To obtain the flux $\mathbf{F}^{\mathrm{num},EC,j}$, we need to write $\llbracket \omega \rrbracket$ and $\llbracket \psi_j \rrbracket$ in terms of jumps our preferred variables (in this case, the primitive variables $\mathbf{z}$)
\begin{equation}
    \llbracket \bm{\omega} \rrbracket = A\left(\mathbf{z}_-, \mathbf{z}_+\right) \cdot \llbracket \mathbf{z} \rrbracket,\:A\left(\mathbf{z}_-, \mathbf{z}_+\right) \in \mathbb{R}^{4 \times 4},
\end{equation}
\begin{equation}
    \llbracket \psi_j \rrbracket = B\left(\mathbf{z}_-, \mathbf{z}_+\right) \cdot \llbracket \mathbf{z} \rrbracket,\:B\left(\mathbf{z}_-, \mathbf{z}_+\right) \in \mathbb{R}^{1 \times 4}.
\end{equation}
Here $A\left(\mathbf{z}_-, \mathbf{z}_+\right)$ and $B\left(\mathbf{z}_-, \mathbf{z}_+\right)$ are matrices dependent on the values of the primitive variables at the two points between which the flux is computed.
The compatibility condition~(\ref{eqn:ec-fluxes}) can be rewritten as
\begin{equation}
     \left(\left(\mathbf{F}^{\mathrm{num},EC,j}\right)^{\mathrm{T}} \cdot  A\left(\mathbf{z}_-, \mathbf{z}_+\right)  - B\left(\mathbf{z}_-, \mathbf{z}_+\right) \right)  \cdot \llbracket \mathbf{z} \rrbracket = 0,\:\forall\:  \llbracket \mathbf{z} \rrbracket \in \mathbb{R}^4.\label{eqn:ec-fluxes-matrix}
\end{equation}
Once $A\left(\mathbf{z}_-, \mathbf{z}_+\right)$ and $B\left(\mathbf{z}_-, \mathbf{z}_+\right)$ have been found, $\mathbf{F}^{\mathrm{num},EC,j}$ can be computed as
\begin{equation}
    \mathbf{F}^{\mathrm{num},EC,j} = \left(B\left(\mathbf{z}_-, \mathbf{z}_+\right) A^{-1}\left(\mathbf{z}_-, \mathbf{z}_+\right) \right)^{\mathrm{T}}.
\end{equation}

In order to obtain rewrite the jumps in the entropy variables in terms of jumps in the primitive variables, we first need to introduce the following averaging operators~\cite{ismail2009affordable,ranocha2018comparison}
\begin{equation}
    \left\{\!\left\{ a \right\}\!\right\} = \frac{1}{2}\left(a_- + a_+\right),\label{eq:avg_mean}
\end{equation}
\begin{equation}
    \left\{\!\left\{ a \right\}\!\right\}_{\mathrm{geo}} = \sqrt{a_- a_+},\label{eq:avg_geo}
\end{equation}
\begin{equation}
    \left\{\!\left\{ a \right\}\!\right\}_{\mathrm{log}} = \frac{a_+ - a_-}{\log a_+ - \log a_-}.\label{eq:avg_log}
\end{equation}
Then, the following relations hold:
\begin{equation}
    \llbracket ab \rrbracket =  \left\{\!\left\{ a \right\}\!\right\} \llbracket b \rrbracket + \left\{\!\left\{ b \right\}\!\right\} \llbracket a \rrbracket,
\end{equation}
\begin{equation}
    \left \llbracket \frac{a}{b} \right\rrbracket = -\frac{\left\{\!\left\{ a \right\}\!\right\}} {\left\{\!\left\{ b \right\}\!\right\}_{\mathrm{geo}}^2} \llbracket b \rrbracket +  \left\{\!\left\{ \frac{1}{b} \right\}\!\right\} \llbracket a \rrbracket,
\end{equation}
\begin{equation}
    \llbracket \log a \rrbracket =  \frac{\llbracket a \rrbracket} {\left\{\!\left\{ a \right\}\!\right\}_{\mathrm{log}}} .
\end{equation}


Let us write out the jumps for $\omega_2$, $\omega_3$, $\omega_4$ (as defined by~(\ref{eq:omega-vector})), as their derivation is straightforward and does not require any additional discussion. We find
\begin{equation}
    \left \llbracket \omega_2 \right\rrbracket =
    \left\{\!\left\{ v_x \right\}\!\right\}  \left \llbracket \frac{1}{T} \right\rrbracket + \left\{\!\left\{ \frac{1}{T} \right\}\!\right\} \left \llbracket v_x \right\rrbracket =
    -\left\{\!\left\{ v_x \right\}\!\right\} \frac{1}{{\left\{\!\left\{ T \right\}\!\right\}_{\mathrm{geo}}^2}}  \left \llbracket T \right\rrbracket + \left\{\!\left\{ \frac{1}{T} \right\}\!\right\} \left \llbracket v_x \right\rrbracket.
\end{equation}
By analogy we obtain
\begin{equation}
    \left \llbracket \omega_3 \right\rrbracket =
    -\left\{\!\left\{ v_y \right\}\!\right\} \frac{1}{{\left\{\!\left\{ T \right\}\!\right\}_{\mathrm{geo}}^2}}  \left \llbracket T \right\rrbracket + \left\{\!\left\{ \frac{1}{T} \right\}\!\right\} \left \llbracket v_y \right\rrbracket.
\end{equation}
Finally, for $\left\llbracket \omega_4 \right\rrbracket$ we get
\begin{equation}
    \left \llbracket \omega_4 \right\rrbracket =
    \frac{1}{{\left\{\!\left\{ T \right\}\!\right\}_{\mathrm{geo}}^2}}  \left \llbracket T \right\rrbracket.
\end{equation}

Let us now focus the attention on $\omega_1$, taking into account the definition of entropy~(\ref{eq:entropy-def}).
\begin{align}
    \left \llbracket \omega_1 \right\rrbracket & = -\left \llbracket s \right\rrbracket + \left \llbracket \frac{\varepsilon_{\mathrm{int}}}{T} \right\rrbracket - \left \llbracket \frac{v_x^2 + v_y^2}{2T} \right\rrbracket \nonumber \\
    & = -\left \llbracket s \right\rrbracket - \frac{\left\{\!\left\{ \varepsilon_{\mathrm{int}} \right\}\!\right\}} {\left\{\!\left\{ T \right\}\!\right\}_{\mathrm{geo}}^2} \llbracket T \rrbracket + \left\{\!\left\{ \frac{1}{T} \right\}\!\right\} \left \llbracket \varepsilon_{\mathrm{int}} \right\rrbracket + \frac{\left\{\!\left\{ v_x^2 \right\}\!\right\} +\left\{\!\left\{ v_y^2 \right\}\!\right\}} {2 \left\{\!\left\{ T \right\}\!\right\}_{\mathrm{geo}}^2} \llbracket T \rrbracket \nonumber \\
    &- \left\{\!\left\{ \frac{1}{T} \right\}\!\right\} \left\{\!\left\{ v_x \right\}\!\right\}  \llbracket v_x \rrbracket - \left\{\!\left\{ \frac{1}{T} \right\}\!\right\} \left\{\!\left\{ v_y \right\}\!\right\}  \llbracket v_y \rrbracket \nonumber \\
    &= - \left\llbracket \eta(T) \right\rrbracket  + 
    \frac{k}{m} \frac{\llbracket \rho \rrbracket} {\left\{\!\left\{ \rho \right\}\!\right\}_{\mathrm{log}}} - \frac{\left\{\!\left\{ \varepsilon_{\mathrm{int}} \right\}\!\right\}} {\left\{\!\left\{ T \right\}\!\right\}_{\mathrm{geo}}^2} \llbracket T \rrbracket + \left\{\!\left\{ \frac{1}{T} \right\}\!\right\} \left \llbracket \varepsilon_{\mathrm{int}} \right\rrbracket \nonumber\\
    &+ \frac{\left\{\!\left\{ v_x^2 \right\}\!\right\} +\left\{\!\left\{ v_y^2 \right\}\!\right\}} {2 \left\{\!\left\{ T \right\}\!\right\}_{\mathrm{geo}}^2} \llbracket T \rrbracket - \left\{\!\left\{ \frac{1}{T} \right\}\!\right\} \left\{\!\left\{ v_x \right\}\!\right\}  \llbracket v_x \rrbracket - \left\{\!\left\{ \frac{1}{T} \right\}\!\right\} \left\{\!\left\{ v_y \right\}\!\right\}  \llbracket v_y \rrbracket.\label{eq:jump-omega1}
\end{align}
All that is left is to express the jumps in $\eta(T)$ and $\varepsilon_{\mathrm{int}}(T)$ in terms of jumps in $T$. At first this seems impossible unless some specific expressions for $c_v(T)$ and $\varepsilon_{\mathrm{int}}$ are provided. Even then, a closed-form solution is not necessarily available. For the simplest case of an infinite harmonic oscillator~\cite{peyvan2023high}, the resulting flux formulae are quite complex and involve hyperbolic functions. Another approach is to use the NASA polynomials~\cite{mcbride2002nasa}, which approximate specific heats as polynomials in $T$ and $1/T$. The polynomial fits for entropies and enthalpies are obtained via simple analytical integration. One can plug in the polynomial approximations into~(\ref{eq:jump-omega1}) and obtain a closed-form expression for the jump in $\omega_1$ using the relevant chain rules, which is the procedure employed in~\cite{gouasmi2020formulation}. The drawback of the approach is that for each internal energy model, a new set of fits has to be performed as the coefficients given by NASA take into account effects such as anharmonicity, electronic excitation, etc.  Moreover, in thermally non-equilibrium multi-temperature flows, additional fits for the vibrational components of the internal energies need to be provided.

However, we can formally write out the jumps in $\eta(t)$ and $\varepsilon_{\mathrm{int}}$ using the mean value theorem:
\begin{equation}
   \left \llbracket \eta(T) \right \rrbracket = \int_0^{T_+} \frac{c_v(\tau)}{\tau}\mathrm{d}\tau - \int_0^{T_-} \frac{c_v(\tau)}{\tau}\mathrm{d}\tau = \int_{T_-}^{T_+} \frac{c_v(\tau)}{\tau}\mathrm{d}\tau = \llbracket T \rrbracket \frac{c_v\left(T^\ast\right)}{T^\ast},\label{eq:jump-s}
\end{equation}
\begin{equation}
   \left \llbracket \varepsilon_{\mathrm{int}} \right \rrbracket = \int_0^{T_+} c_v(\tau)\mathrm{d}\tau - \int_0^{T_-} c_v(\tau)\mathrm{d}\tau = \int_{T_-}^{T_+} c_v(\tau) \mathrm{d}\tau = \llbracket T \rrbracket c_v\left(T^{\ast\ast}\right).\label{eq:jump-eint}
\end{equation}
Here $T^\ast$, $T^{\ast\ast}$ are some values of the temperature in the interval $\left [\mathrm{min}\left(T_-, T_+\right), \mathrm{max}\left(T_-, T_+\right) \right]$. Based on this formalism, we can simply define the following quantities:
\begin{equation}
    \frac{c_v\left(T^\ast\right)}{T^\ast} \coloneq \frac{\left \llbracket \eta(T) \right \rrbracket}{\llbracket T \rrbracket},
\end{equation}
\begin{equation}
    c_v\left(T^{\ast\ast}\right) \coloneq \frac{\left \llbracket \varepsilon_{\mathrm{int}} \right \rrbracket}{\llbracket T \rrbracket}.
\end{equation}


With that, we can write out the final expression for $\left \llbracket \omega_1 \right\rrbracket$:
\begin{align}
    \left \llbracket \omega_1 \right\rrbracket & = - \llbracket T \rrbracket \frac{c_v\left(T^\ast\right)}{T^\ast}  + 
    \frac{k}{m} \frac{\llbracket \rho \rrbracket} {\left\{\!\left\{ \rho \right\}\!\right\}_{\mathrm{log}}} - \frac{\left\{\!\left\{ \varepsilon_{\mathrm{int}} \right\}\!\right\}} {\left\{\!\left\{ T \right\}\!\right\}_{\mathrm{geo}}^2} \llbracket T \rrbracket + \left\{\!\left\{ \frac{1}{T} \right\}\!\right\} c_v\left(T^{\ast\ast}\right)\llbracket T \rrbracket  \nonumber\\
    &+ \frac{\left\{\!\left\{ v_x^2 \right\}\!\right\} +\left\{\!\left\{ v_y^2 \right\}\!\right\}} {2 \left\{\!\left\{ T \right\}\!\right\}_{\mathrm{geo}}^2} \llbracket T \rrbracket - \left\{\!\left\{ \frac{1}{T} \right\}\!\right\} \left\{\!\left\{ v_x \right\}\!\right\}  \llbracket v_x \rrbracket - \left\{\!\left\{ \frac{1}{T} \right\}\!\right\} \left\{\!\left\{ v_y \right\}\!\right\}  \llbracket v_y \rrbracket.\label{eq:jump-omega1-final}
\end{align}

Substituting into~(\ref{eqn:ec-fluxes}) and solving for the flux components, we obtain two entropy-conserving flux vectors (in the $x$ and $y$ spatial directions) for the conservative variables:
\begin{align}
    F^{\mathrm{num},x}_{\rho} &  = {\left\{\!\left\{ \rho \right\}\!\right\}_{\mathrm{log}}}\left\{\!\left\{ v_x \right\}\!\right\},\label{eq:flux_fx_rho} \\
    F^{\mathrm{num},x}_{\rho v_x} &  =  F^{\mathrm{num},x}_{\rho} \left\{\!\left\{ v_x \right\}\!\right\}  + \frac{k}{m}\frac{ \left\{\!\left\{ \rho \right\}\!\right\}}{ \left\{\!\left\{ 1/T \right\}\!\right\}}, \\
    F^{\mathrm{num},x}_{\rho v_y} & = F^{\mathrm{num},x}_{\rho} \left\{\!\left\{ v_y \right\}\!\right\}, \\
    F^{\mathrm{num},x}_{E} & = F^{\mathrm{num},x}_{\rho} \left\{\!\left\{ T \right\}\!\right\}_{\mathrm{geo}}^2 \left(\frac{c_v\left(T^\ast\right)}{T^\ast} -\left\{\!\left\{ \frac{1}{T} \right\}\!\right\} c_v\left(T^{\ast\ast}\right) \right) \nonumber \\
    & + F^{\mathrm{num},x}_{\rho}  \left(\left\{\!\left\{ \varepsilon_{\mathrm{int}} \right\}\!\right\} - \frac{\left\{\!\left\{ v_x^2 \right\}\!\right\} +\left\{\!\left\{ v_y^2 \right\}\!\right\}} {2}\right) \nonumber \\
    & + \left\{\!\left\{ v_x \right\}\!\right\} F^{\mathrm{num},x}_{\rho v_x} +  \left\{\!\left\{ v_y \right\}\!\right\}  F^{\mathrm{num},x}_{\rho v_y}.\label{eq:flux-E_x}
\end{align}

In the $y$ direction:
\begin{align}
    F^{\mathrm{num},y}_{\rho} & = {\left\{\!\left\{ \rho \right\}\!\right\}_{\mathrm{log}}}\left\{\!\left\{ v_y \right\}\!\right\},\\
    F^{\mathrm{num},y}_{\rho v_x} & =  F^{\mathrm{num},y}_{\rho} \left\{\!\left\{ v_x \right\}\!\right\},\\
    F^{\mathrm{num},y}_{\rho v_y} & = F^{\mathrm{num},y}_{\rho} \left\{\!\left\{ v_y \right\}\!\right\}  + \frac{k}{m} \frac{ \left\{\!\left\{ \rho \right\}\!\right\}}{ \left\{\!\left\{ 1/T \right\}\!\right\}},\\
    F^{\mathrm{num},y}_{E} & = F^{\mathrm{num},y}_{\rho} \left\{\!\left\{ T \right\}\!\right\}_{\mathrm{geo}}^2 \left(\frac{c_v\left(T^\ast\right)}{T^\ast} -\left\{\!\left\{ \frac{1}{T} \right\}\!\right\} c_v\left(T^{\ast\ast}\right) \right) \nonumber \\
    & + F^{\mathrm{num},y}_{\rho}  \left(\left\{\!\left\{ \varepsilon_{\mathrm{int}} \right\}\!\right\} - \frac{\left\{\!\left\{ v_x^2 \right\}\!\right\} +\left\{\!\left\{ v_y^2 \right\}\!\right\}} {2}\right) \nonumber \\
    & + \left\{\!\left\{ v_x \right\}\!\right\} F^{\mathrm{num},y}_{\rho v_x} +  \left\{\!\left\{ v_y \right\}\!\right\}  F^{\mathrm{num},y}_{\rho v_y}.\label{eq:flux-ec-energy}
\end{align}

The only remaining part is the computation of $\frac{c_v\left(T^\ast\right)}{T^\ast}$, $c_v\left(T^{\ast\ast}\right)$. It should be stressed that up to this point, no quantities have been approximated.
For practical purposes, however, one requires some numerical procedures to compute $\llbracket T \rrbracket$ and $\left \llbracket \eta(T) \right \rrbracket$. For example, given an analytical expression for $\varepsilon_{\mathrm{int}}(T)$ (e.g. one incorporating the infinite harmonic oscillator model~(\ref{eq:iho}) or cut-off oscillator model~(\ref{eq:cutoffosc}) for the internal energy), one can use Newton's method to compute $T_-$, $T_+$ from the given values of $\varepsilon_{\mathrm{int},-}$, $\varepsilon_{\mathrm{int},+}$ to a desired degree of precision. From the computed temperature values, one can then evaluate $\eta(T_-)$ and $\eta(T_+)$ numerically using some quadrature rule. The disadvantage of such an approach is the potentially high computational cost, as each flux evaluation requires multiple evaluations of $\varepsilon_{\mathrm{int}}(T)$ and $c_v(T)$ in the Newton solver iterations and the numerical integration for $\eta(T)$.

In the present work, we propose the following procedure. Values of $\varepsilon_{\mathrm{int}}$ and $c_v(T)$ are pre-computed and tabulated at the start of a simulation over a range of temperatures with a uniform step size of $\Delta T$. During the course of a simulation, piece-wise linear interpolation is then used to compute the values of $\varepsilon_{\mathrm{int}}$, $c_v(T)$ where required. The piece-wise linear interpolation is of $c_v(T)$ is used to compute $\eta(T)$ exactly. We can write
\begin{equation}
    \eta(T) = \sum_{i=0}^{N-1} \eta_i + \left(c_{v,N} - \frac{(c_v(T) - c_{v,N}) T_{N}}{\Delta T}\right) \ln\left(\frac{T}{T_{N}}\right) + \left(c_v(T) - c_{v,N}\right).\label{eq:eta_compute}
\end{equation}
Here $c_{v,i}$ denotes a tabulated value of $c_v(T)$ computed at a tabulated temperature of $T_i$, $c_v(T)$ is a linearly interpolated value of $c_v$ at an arbitrary temperature. $N$ is defined as $\lfloor T - T_{\mathrm{min}} / \Delta T \rfloor$, where $T_{min}$ is the minimum temperature used for pre-computing the values. The values $\eta_i$ are the integrals of $c_v(T)/T$ computed over the $\Delta T$-sized intervals assuming a linear interpolation of $c_v(T)$ over the interval:
\begin{equation}
    \eta_i = \left(c_{v,i} - \frac{(c_{v,i+1} - c_{v,i}) T_{i}}{\Delta T}\right) \ln\left(\frac{T_{i+1}}{T_{i}}\right) + \left(c_{v,i+1} - c_{v,i}\right).
\end{equation}

To reduce the computational cost of computing $\eta(T)$ during a simulation, the values of $\sum_{i=0}^{N} \eta_i$ can also be tabulated for $N=0,\ldots N_{\mathrm{max}}$, and the computation of $\eta(T)$ is then reduced to getting a value from a look-up table and evaluating one logarithm. The proposed algorithm not only provides a direct numerical procedure for evaluation of the entropy-conservative flux presented above for a gas with an arbitrary equation for the internal energy distribution function, but also potentially leads to a reduction in the computational cost (as potentially expensive on-the-fly computations are replaced by linear interpolation). In principle, other higher-order interpolation methods can be used, as the developed flux formulation is independent of the numerical procedures used to estimate $\frac{c_v\left(T^\ast\right)}{T^\ast}$, $c_v\left(T^{\ast\ast}\right)$. The error due to the interpolation of energy and specific heats will be analyzed in the next section, and in the numerical results section it will be shown that with a reasonable choice of $\Delta T$ the error in the entropy production rate can be made to be on the order of machine precision, thus guaranteeing numerical entropy conservation.

This concludes the derivation of the entropy-conservative numerical flux for the case of a single-species gas with arbitrary internal energies.

\subsection{Properties of the entropy-conservative flux}
We now investigate and prove several properties of the flux~(\ref{eq:flux_fx_rho})--(\ref{eq:flux-ec-energy}).

\begin{theorem}
The numerical flux given by Eqns.~(\ref{eq:flux_fx_rho})--(\ref{eq:flux-ec-energy}) is consistent if $\lim_{\substack{T_- \to T \\ T_+ \to T}} {c_v\left(T^\ast\right)}/{T^\ast} = {c_v(T)}/{T}$
and $\lim_{\substack{T_- \to T \\ T_+ \to T}}c_v\left(T^{\ast\ast}\right) = c_v(T)$.\label{theorem:consistency}
\end{theorem}
\begin{proof}
We prove consistency of the flux in the $x$ spatial direction, as consistency of the flux in the $y$ direction can be proved by analogy.
We need to show that $\mathbf{F}^{\mathrm{num},x}\left(\mathbf{u},\mathbf{u}\right)=\mathbf{f}_x(\mathbf{u})$.

The averaging operators~(\ref{eq:avg_mean})--(\ref{eq:avg_log}) are consistent, thus immediately leading to consistency of the density and momentum fluxes ${F}^{\mathrm{num},x}_{\rho}$, ${F}^{\mathrm{num},x}_{\rho v_x}$, ${F}^{\mathrm{num},x}_{\rho v_y}$.

We now consider the energy flux in more detail. Per the requirement that ${c_v\left(T^\ast\right)}/{T^\ast} \to {c_v(T)}/{T}$ and $c_v\left(T^{\ast\ast}\right) \to c_v(T)$ as $T_+ \to T$ and $T_- \to T$ we have that 
${c_v\left(T^\ast\right)}/{T^\ast} -\left\{\!\left\{ 1 / T \right\}\!\right\} c_v\left(T^{\ast\ast}\right) = 0$.

Using the consistency of the averaging operators we can therefore write 
\begin{equation}
F^{\mathrm{num},x}_{E} = F^{\mathrm{num},x}_{\rho}  \left(\varepsilon_{\mathrm{int}} - \frac{v_{x}^2 + v_{y}^2 } {2}\right) +  v_{x} F^{\mathrm{num},x}_{\rho v_x} +  v_{y}   F^{\mathrm{num},x}_{\rho v_y}.
\end{equation}
Applying the consistency of the density and momentum fluxes, the energy flux can be rewritten as
\begin{equation}
F^{\mathrm{num},x}_{E} = \rho v_{x} \left(\varepsilon_{\mathrm{int}} - \frac{v_{x}^2 + v_{y}^2 } {2}\right) +  \rho v_{x}^3 + v_{x} p +  \rho v_{x} v_{y}^2 = v_{x}\left(\rho \varepsilon_{\mathrm{int}} + \rho  \frac{v_{x}^2 + v_{y}^2 }{2} + p \right) = v_x \left(E+p \right).
\end{equation}

This can be seen to be exactly equal to the energy flux in the $x$ direction for Euler equations as given by~(\ref{eqns:euler-flux_x}).
\end{proof}

\begin{theorem}
The numerical flux given by Eqns.~(\ref{eq:flux_fx_rho})--(\ref{eq:flux-ec-energy}) is kinetic energy preserving.
\end{theorem}
\begin{proof}
We consider the flux in the $x$ direction, proof for the $y$ direction is analogous. The momentum flux for $x$ component of the momentum is given by
\begin{equation}
    F^{\mathrm{num},x}_{\rho v_x} =  F^{\mathrm{num},x}_{\rho} \left\{\!\left\{ v_x \right\}\!\right\}  + \frac{k}{m}\frac{ \left\{\!\left\{ \rho \right\}\!\right\}}{ \left\{\!\left\{ 1/T \right\}\!\right\}}.
\end{equation}
The quantity $\frac{k}{m}\frac{ \left\{\!\left\{ \rho \right\}\!\right\}}{ \left\{\!\left\{ 1/T \right\}\!\right\}}$ is a consistent approximation of the pressure, and thus, the numerical flux~(\ref{eq:flux_fx_rho})--(\ref{eq:flux-ec-energy}) is kinetic energy preserving~\cite{jameson2008formulation,chandrashekar2013kinetic,ranocha2018comparison}.
\end{proof}

\begin{remark}
In the case of $c_v=\mathrm{const}$, the flux given by Eqns.~(\ref{eq:flux_fx_rho})--(\ref{eq:flux-ec-energy}) and evaluated using the proposed procedure (piece-wise linear interpolation of $c_v(T)$ and subsequent use of Eqn.~\ref{eq:eta_compute}) reduces to the flux proposed by Chandrashekar~\cite{chandrashekar2013kinetic}, differing only up to a trivial change of variables from $T$ to $R_{\mathrm{specific}}T$.
\end{remark}

\begin{proof}
For the density and velocity fluxes, this is immediately evident if one uses $R_{\mathrm{specific}}T$ instead of $T$ as the variable, where $R_{\mathrm{specific}}=\frac{k}{m}$.
All that remains is to simplify the expressions for the energy flux (we only consider the flux in the $y$ direction, the derivation for the $x$ direction is done by analogy).
For constant $c_v$, we have that $\eta(T_+) - \eta(T_-) = \int_{T_-}^{T_+} {c_v}/{\tau}\mathrm{d}\tau=c_v \left(\log T_+ - \log T_-\right)$.  Since this is a specific case of a piece-wise linear $c_v$, this analytical integration is recovered exactly in our proposed numerical procedure. We can therefore write that $T^{\ast} = \left\{\!\left\{ T \right\}\!\right\}_{\mathrm{log}}$. We also have that $\left\{\!\left\{ \varepsilon_{\mathrm{int}} \right\}\!\right\} = c_v \left\{\!\left\{ T \right\}\!\right\}$. Therefore, we can re-write expression~(\ref{eq:flux-ec-energy}) as
\begin{align}
    F^{\mathrm{num},y}_{E} & = F^{\mathrm{num},y}_{\rho} c_v\left(\frac{\left\{\!\left\{ T \right\}\!\right\}_{\mathrm{geo}}^2}{\left\{\!\left\{ T \right\}\!\right\}_{\mathrm{log}}} - \left\{\!\left\{ T \right\}\!\right\}_{\mathrm{geo}}^2 \left\{\!\left\{ \frac{1}{T} \right\}\!\right\} + \left\{\!\left\{ T \right\}\!\right\} \right) \nonumber \\
    & - F^{\mathrm{num},y}_{\rho}  \frac{\left\{\!\left\{ v_x^2 \right\}\!\right\} +\left\{\!\left\{ v_y^2 \right\}\!\right\}} {2} \nonumber \\
    & + \left\{\!\left\{ v_x \right\}\!\right\} F^{\mathrm{num},y}_{\rho v_x} +  \left\{\!\left\{ v_y \right\}\!\right\}  F^{\mathrm{num},y}_{\rho v_y}.\label{eq:flux-ec-energy-2}
\end{align}
The term $- \left\{\!\left\{ T \right\}\!\right\}_{\mathrm{geo}}^2 \left\{\!\left\{ \frac{1}{T} \right\}\!\right\} + \left\{\!\left\{ T \right\}\!\right\}$ can easily be shown to be equal to 0. Finally, using the definition of $\gamma = c_p/c_v$ and Mayer's relation $c_p = c_v + k/m$, we can write $k/(m(\gamma - 1))$ instead of $c_v$. Thus, we obtain the final expression for the energy flux:
\begin{align}
    F^{\mathrm{num},y}_{E} & = F^{\mathrm{num},y}_{\rho} \left(\frac{k}{m} \frac{1}{\gamma-1}\frac{\left\{\!\left\{ T \right\}\!\right\}_{\mathrm{geo}}^2}{\left\{\!\left\{ T \right\}\!\right\}_{\mathrm{log}}}  - \frac{\left\{\!\left\{ v_x^2 \right\}\!\right\} +\left\{\!\left\{ v_y^2 \right\}\!\right\}} {2} \right) \nonumber \\
    & + \left\{\!\left\{ v_x \right\}\!\right\} F^{\mathrm{num},y}_{\rho v_x} +  \left\{\!\left\{ v_y \right\}\!\right\}  F^{\mathrm{num},y}_{\rho v_y}.\label{eq:flux-ec-energy-3}
\end{align}
This can be seen to be exactly the energy flux as proposed by Chandrashekar~\cite{chandrashekar2013kinetic} up to a trivial change of variables from $T$ to $R_{\mathrm{specific}}T$.
\end{proof}


\begin{theorem}
    Let $\varepsilon_{\mathrm{int}}(T)$ and $c_v(T)$ be piece-wise interpolated with error of order $(\Delta T)^n$, where $\Delta T$ is the spacing between equidistant values of $T$ used for the interpolation. We assume $\eta(T)$ to be computed exactly using the piece-wise interpolation for $c_v(T)$ (see Eqn.~(\ref{eq:eta_compute})).
    Let us also assume 1) a negligibly small error in the non-linear estimation of temperature from the interpolated energy, 2) that the function $\varepsilon_{\mathrm{int}}(T)$ is differentiable, 3) the derivative $c_v(T) = \varepsilon_{\mathrm{int}}'(T)$ is bounded from below by a non-zero value.
    Then the the error in the numerical energy flux evaluated using Eqn.~(\ref{eq:flux-E_x}) using a piece-wise interpolation for $\varepsilon_{\mathrm{int}}(T)$ and $c_v(T)$ when compared to an exact computation of the flux (that is, it is assumed that one can compute $T(\varepsilon_{\mathrm{int}})$ and $\eta(T)$ with no numerical errors) is of the order $\left(\Delta T\right)^{n}$.
\end{theorem}

First, we prove the following lemma. 
\begin{lemma}
Let $p(x)$ be a piece-wise interpolation of a differentiable function $f(x)$ with error of order $(\Delta x)^n$, where $\Delta x$ is the spacing between equidistant values of $x$ used for the interpolation. Let $f'(x) \ge A > 0$, where $A$ is a constant. Then given a value $f_0=f(x_0)$, the following holds: $\left|f^{-1}(f_0) - p^{-1}(f_0) \right| \leq C (\Delta x)^n$, where $C$ is a constant. \label{lemma-inverse}
\end{lemma}
\begin{proof}
Let us denote $p^{-1}(f_0)$ as $\tilde{x}_0$. As $f$ is differentiable, we can apply the mean value theorem, assuming without loss of generality that $\tilde{x}_0 > x_0$:
\begin{equation}
    \frac{f(\tilde{x}_0) - f(x_0)}{\tilde{x}_0 - x_0} = f'(c),\: c \in [x_0,\tilde{x}_0].\label{eq:mvt}
\end{equation}
We have that $f(x_0) = f_0 = p(\tilde{x}_0)$, so we can write that 
\begin{equation}
    \left| \tilde{x}_0 - x_0 \right| =  \frac{\left|f(\tilde{x}_0) - f(x_0)\right|}{\left|f'(c)\right|}=\frac{\left|f(\tilde{x}_0) -  p(\tilde{x}_0)\right|}{\left|f'(c)\right|} \leq \frac{C_1 (\Delta x)^n} {\left|f'(c)\right|} \leq \frac{C_1}{A} (\Delta x)^n = C (\Delta x)^n.
\end{equation}
\end{proof}

Now we prove the statement of the theorem.
\begin{proof}
We consider the energy flux in the $x$ direction, as proof for the $y$ direction is done by analogy.
Given two vector of conservative variables $\mathbf{u}_{-}$, $\mathbf{u}_{+}$, the values of $\rho$, $v_x$, $v_y$, and $\varepsilon_{\rm{int}}$ for the states $-$ and $+$ can be computed exactly.

Therefore the error comes only from the term $\left\{\!\left\{ T \right\}\!\right\}_{\mathrm{geo}}^2 \left(\frac{c_v\left(T^\ast\right)}{T^\ast} -\left\{\!\left\{ \frac{1}{T} \right\}\!\right\} c_v\left(T^{\ast\ast}\right) \right)$.

Recalling the definitions of ${c_v\left(T^\ast\right)} / {T^\ast}$, $c_v\left(T^{\ast\ast}\right)$, we can re-write the above term as
\begin{equation}
    \left\{\!\left\{ T \right\}\!\right\}_{\mathrm{geo}}^2 \left(\frac{\eta(T_+) - \eta(T_-)}{T_+ - T_-} -\left\{\!\left\{ \frac{1}{T} \right\}\!\right\} \frac{\varepsilon_{\rm{int},+} - \varepsilon_{\rm{int},-}}{T_+ - T_-} \right).\label{eq:flux_error_source}
\end{equation}

According to Lemma~\ref{lemma-inverse}, the error in the estimation of $T_{-}$ and $T_{+}$ is of order $(\Delta T)^n$. Therefore, the errors in the computation of $\left\{\!\left\{ T \right\}\!\right\}_{\mathrm{geo}}^2$, $\left\{\!\left\{ \frac{1}{T} \right\}\!\right\}$, and $\frac{1}{T_+ - T_-}$ are also all of order $(\Delta T)^n$; this also holds for the error in the computation of $\varepsilon_{\rm{int},+}$, $\varepsilon_{\rm{int},-}$ and for any products of these quantities.
Therefore, the error in $\left\{\!\left\{ T \right\}\!\right\}_{\mathrm{geo}}^2 \left\{\!\left\{ \frac{1}{T} \right\}\!\right\} \frac{\varepsilon_{\rm{int},+} - \varepsilon_{\rm{int},-}}{T_+ - T_-}$ is of order $(\Delta T)^n$.

The only remaining term to analyze is $\eta(T_+) - \eta(T_-)$. There are two sources of error: 1) the error due to the interpolation used for $c_v(T)$ and 2) the error due to the choice of interpolation points $T_-$, $T_+$ (as they are computed from an interpolation of energy). The errors due to the interpolation of $c_v(T)$ and the choice of integration limits are both of order $(\Delta T)^n$. Thus, the errors in the computation of $\eta(T_+)$,  $\eta(T_-)$ are also of order $(\Delta T)^n$

Therefore, the term~(\ref{eq:flux_error_source}) appearing in the energy flux has error of order $(\Delta T)^n$.
\end{proof}

This concludes the analysis of the properties of the developed entropy-conserving numerical fluxes for a gas with an arbitrary internal energy function. To summarize:
\begin{itemize}
    \item The flux is consistent.
    \item The flux is kinetic energy preserving.
    \item The flux reduces to the entropy-conservative flux of Chandrashekar in case of a constant specific heat.
    \item If piece-wise interpolation with a step-size of $\Delta T$ and error of order $(\Delta T)^n$ is used to compute the internal energy $\varepsilon_{\mathrm{int}}(T)$ and $c_v(T)$, whilst the integral part of the entropy $\eta$ is computed exactly using the piece-wise interpolation for $c_v(T)$, then the error in the energy flux is of the order $(\Delta T)^n$.
\end{itemize}

\subsection{Algorithm for pre-processing and flux computation}
Here, an outline of the overall algorithm for the pre-computation of relevant tabulated values and subsequent estimation of the flux components is presented.

\begin{algorithm}
\caption{Pre-computation of tabulated quantities}\label{alg:tables}
\begin{algorithmic}
\Require $T_{min} > 0$,  $T_{max} > T_{min}$, $\Delta T > 0$, $\varepsilon_{\mathrm{int}}(T)$, $c_v(T)$
\State $N \gets \lfloor (T_{max} - T_{min}) / \Delta T \rfloor$  \Comment{Compute number of array elements}

\For{$i = 0, \dots, N$}
    \State $T_A[i] \gets T_{min} + i \Delta T$ \Comment{Set element of uniformly spaced temperature array}
    
    
    \State $E_A[i] \gets \varepsilon_{\mathrm{int}}(T_A[i])$ \Comment{Compute and set element of energy array}
    \State $c_{v,A}[i] \gets c_{v}(T_A[i]) $ \Comment{Compute and set element of specific heat array}
    
    \If{$i$ is 0}
    \State $\eta_{A}[i] \gets 0$
    \Else
    \State $\eta_{A}[i] \gets \eta_{A}[i-1] + \left(c_{v,A}[i-1] - \frac{(c_{v,A}[i] - c_{v,A}[i-1]) T_A[i]}{\Delta T}\right) \ln\left(\frac{T_A[i]}{T_A[i-1]}\right) + \left(c_{v,A}[i] - c_{v,A}[i-1]\right)$    \Comment{Compute integral part of entropy using piece-wise interpolation for $c_v$}
    \EndIf
\EndFor
\end{algorithmic}
\end{algorithm}

Algorithm~\ref{alg:tables} shows a pseudocode description of how the necessary tabulated quantities are precomputed. To distinguish between the flow variables arrays of tabulated values, we subscript the latter with $A$. All of the arrays $T_A$, $E_{A}$, $c_{v,A}$, $\eta_A$ are of size $N+1$ and store the tabulated values. The user has to provide the minimum and maximum values of the temperature range, as well as the temperature step $\Delta T$. The functions for computation of the specific internal energy $\varepsilon_{\mathrm{int}}$ and the specific heat $c_{v}$ from temperature are assumed to be known. Naturally, the discretization step $\Delta T$ should be chosen such that the error due to the approximations used is acceptably small; the impact of the choice of  $\Delta T$ is discussed in the next section.

\begin{algorithm}
\caption{Computation of quantities necessary for calculation of fluxes}\label{alg:fluxes}
\begin{algorithmic}
\Require $\rho_{-}$, $v_{x,-}$, $v_{y,-}$, $\varepsilon_{\mathrm{int},-}$, $\rho_{+}$, $v_{x,+}$, $v_{y,+}$, $\varepsilon_{\mathrm{int},+}$

\For{$i = -, +$}  \Comment{Compute relevant quantities for left and right states}
    \State $T_{i} \gets T(\varepsilon_{\mathrm{int},i})$ \Comment{Compute temperature from energy using a non-linear solver}
    \State $f_{T,i} \gets (T_{i} - T_{min}) / \Delta T$
    \State $I_{T,i} \gets \lfloor f_{T,i} \rfloor$ \Comment{Find position in array of temperature}
    \State $f_{T,i} \gets f_{T,i} - I_{T,i}$ \Comment{Find linear interpolation coefficient}
    \State $c_{v,i} \gets (1 - f_{T,i}) \cdot c_{v,A}[I_{T,i}] + f_{T,i} \cdot c_{v,A}[I_{T,i} + 1] $  \Comment{Find specific heat from temperature}
    \State $\eta_{i} \gets \eta_{A}[I_{T,i}] + \left(c_{v,A}[I_{T,i}] - \frac{(c_{v,i} - c_{v,A}[I_{T,i}]) T_A[I_{T,i}]}{\Delta T}\right) \ln\left(\frac{T}{T_A[I_{T,i}]}\right) + \left(c_{v,i} - c_{v,A}[I_{T,i}]\right)$  \Comment{Compute integral part of entropy}
\EndFor
\If{$|T_{+} - T_{-}| < r$} \Comment{$r$ is a small fixed tolerance to avoid division by 0}
    \State $\overline{T} \gets \left(T_{+} - T_{-} \right)/2$
    \State $f_{\overline{T}} \gets (\overline{T} - T_{min}) / \Delta T$
    \State $I_{\overline{T}} \gets \lfloor f_{\overline{T}} \rfloor$
    \State $f_{\overline{T}} \gets f_{\overline{T}} - I_{\overline{T}}$
    \State $\overline{c_{v}} \gets (1 - f_{\overline{T}}) \cdot c_{v,A}[I_{\overline{T}}] + f_{\overline{T}} \cdot c_{v,A}[I_{\overline{T}} + 1] $  \Comment{Find specific heat from temperature}
    \State $c_{v}(T^{\ast})/T^{\ast} \gets \overline{c_{v}} / \overline{T}$
    \State $c_{v}(T^{\ast \ast}) \gets \overline{c_{v}}$
\Else
    \State $c_{v}(T^{\ast})/T^{\ast} \gets \left(\eta_{+} - \eta_{-}\right) / \left(T_{+} - T_{-} \right)$
    \State $c_{v}(T^{\ast \ast}) \gets \left(\varepsilon_{\mathrm{int},+} - \varepsilon_{\mathrm{int},-}\right) / \left(T_{+} - T_{-} \right)$
\EndIf
\end{algorithmic}
\end{algorithm}
Algorithm~\ref{alg:fluxes} shows the computation of the quantities required for the calculation of the fluxes as given by Eqns.~(\ref{eq:flux_fx_rho})--(\ref{eq:flux-ec-energy}). First, the temperature is computed from the internal energy, and is then used to approximate the specific heats and integral parts of the entropy of the left and right states. In order to avoid division by zero, in case the temperatures of the left and right states are close (as defined by some prescribed tolerance $r$, for example $r=0.5\Delta T$), the jumps in the integral part of the entropy and internal energy are simply replaced by their derivatives evaluated at $T=(T_{-}+T_{+})/2$. Such a choice of $T$ also ensures consistency of the flux in accordance with Theorem~\ref{theorem:consistency}. It should be noted that the non-linear solver used to compute the temperature from energy also uses the the linear interpolation of the energy and specific heats in the iterations.

\subsection{Computational complexity}
We now briefly analyze the computational cost of the algorithm.
Compared to the case of constant $c_v$, for each entropy-conservative flux computation, the following additional computations are required:
\begin{enumerate}
    \item computation of $T_+$, $T_-$ from $\varepsilon_{\mathrm{int},+}$, $\varepsilon_{\mathrm{int},-}$ using a non-linear solver,
    \item calculation of $\left \llbracket \varepsilon_{\mathrm{int}} \right \rrbracket$, $\left \llbracket   \eta(T) \right \rrbracket$,
    \item calculation of ${c_v\left(T^\ast\right)} / {T^\ast}$, $c_v\left(T^{\ast\ast}\right)$.
\end{enumerate}
In the following analysis we disregard the cost of the computation of temperature from energy (as this can be a call to an independent non-linear solver).
For the case of $|T_{+} - T_{-}| < r$, only subtraction, division, and multiplication operations are required. For the case of $|T_{+} - T_{-}|  \geq r$, additional operations are required, as two specific heats need to be linearly interpolated instead of one; in addition, two logarithms need to be computed. Some of the division operations, such as dividing by $\Delta T$, can be replaced by multiplication operations if the quantity $1/\Delta T$ is pre-computed and stored, leading to a lower computational effort~\cite{fog2020instruction}. 
The formulation in~\cite{peyvan2023high} requires at least three exponentiations and computing two logarithms, whereas the present algorithm does not require the exponentiation operations, thus leading to a potentially lower computational cost. Again, it is important to note that in this brief analysis, the cost of computing the temperature from the energy is neglected.
In addition, it should again be stressed that the computational cost of the flux formulation proposed in the present work is independent of the exact expressions for the internal energy and specific heat used in the precomputation step.
Another possible factor influencing the computational efficiency of the proposed algorithm not considered in the present analysis is the role of cache effects. The approach proposed in the present work requires storing several large one-dimensional arrays of floating point numbers, which are being accessed at every step of the simulation, which could lead to frequent loading and eviction of data from the CPU cache. Timing results of the actual implementations of the various flux functions are presented in the next section.

\subsection{Dissipative flux}
The entropy-conservative flux~(\ref{eq:flux_fx_rho})--(\ref{eq:flux-ec-energy}) can be augmented with a dissipation term and used to reduce the oscillations in flows with strong shocks.
In principle, the dissipative term can be chosen such that flux is entropy-stable:
\begin{equation}
    \mathbf{F}^{\mathrm{num},ES,j} = \mathbf{F}^{\mathrm{num},EC,j} - \frac{1}{2} \mathbf{D} \llbracket \bm{\omega} \rrbracket,\: j=x,y.\label{eq:flux_es_diffusive}
\end{equation}
Here $\mathbf{D}$ is a dissipation matrix constructed in a way such that the flux $\mathbf{F}^{\mathrm{num},ES,j}$ is entropy-stable~\cite{derigs2017novel,winters2017uniquely}. The construction of a suitable dissipation operator $\mathbf{D}$ is left for future work. In the present work, for simulations with strong shocks a local Lax-Friedrichs scalar dissipation is added to the entropy-conservative flux:
\begin{equation}
    \mathbf{F}^{\mathrm{num},D,j} = \mathbf{F}^{\mathrm{num},EC,j} - \frac{\lambda}{2}\llbracket \mathbf{u} \rrbracket,\: j=x,y,\label{eq:flux_diffusive}
\end{equation}
where $\mathbf{F}^{\mathrm{num},ES,i}$ is the entropy-stable version of the entropy-conservative flux $\mathbf{F}^{\mathrm{num},EC,j}$, and $\lambda$ is the fastest local wave speed.
This version of the flux is used for the simulations containing strong shocks presented in the next section.

\section{Numerical results}
To verify the developed flux and apply it to simulation of two-dimensional high-speed flows, the developed algorithms for computation of flow properties and fluxes were implemented in Trixi.jl~\cite{schlottkelakemper2020trixi,ranocha2021adaptive}, a modular framework for solving systems of hyperbolic equations using the DGSEM method. The strong stability preserving Runge-Kutta method~\cite{gottlieb2005high} SSPRK43 was used for time integration, as provided by the DifferentialEquations.jl library~\cite{rackauckas2017differentialequations}. The simulation code is publicly available on Github~\cite{oblapenko2024entropyconservativeRepro}.

\subsection{Comparison to exact expressions for fluxes}
First, the developed flux is compared to existing entropy-conserving fluxes for cases where exact expressions for the latter are available. We it compare to the flux of Chandrashekar~\cite{chandrashekar2013kinetic} in the case of a calorically perfect gas, and to the flux of Peyvan et al.~\cite{peyvan2023high} in the case of a gas described by the infinite harmonic oscillator model.
\begin{figure}[h]
    \centering
    \includegraphics[width=0.7\textwidth]{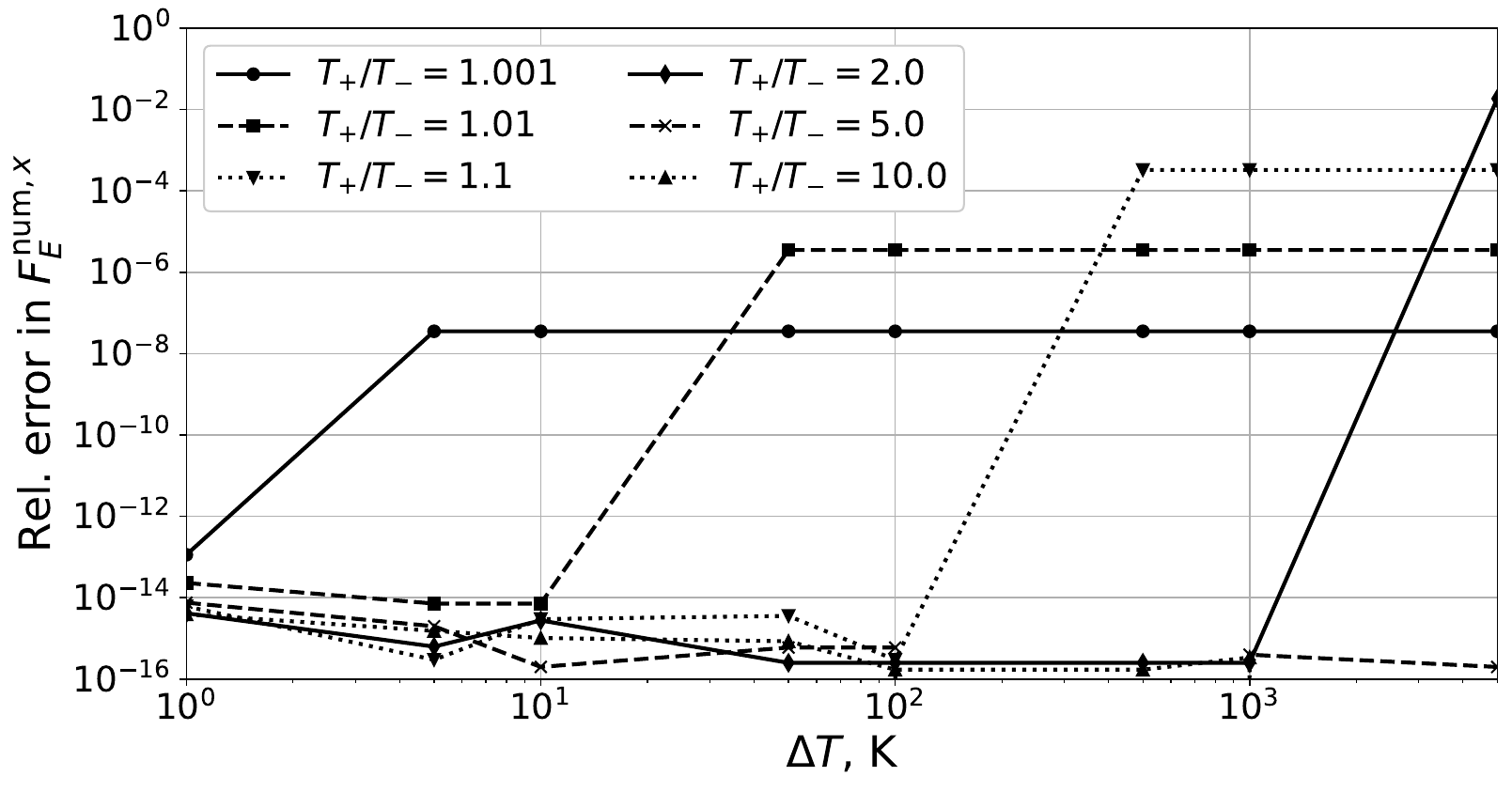}
    \caption{Relative error in the energy flux for different temperature jumps as a function of the discretization step $\Delta T$ for a calorically perfect gas.}
    \label{fig:flux_error}
\end{figure}


\subsubsection{Calorically perfect gas}
In the case of a calorically perfect gas, the error due to interpolation of temperature is almost zero, as in that case temperature is a linear function of energy and the inversion via the Newton-Solver is nearly exact.
For this test case we assume that $c_v=\frac{5}{2}\frac{k}{m}$, which corresponds to an adiabatic index $\gamma=1.4$. The densities and velocities of the left and right states were assumed to be and given by $\rho_{-}=\rho_{+}=3.8485\cdot10^{-3}$~kg/m$^3$, $v_{x,-}=v_{x,+}=1000$~m/s, $v_{y,-}=v_{y,+}=500.0$~m/s, whereas the temperature of the left state $T_{-}$ was taken to be 1000~K and six different temperature jumps were considered: $T_{+}/T_{-}=1.001$, $T_{+}/T_{-}=1.01$, $T_{+}/T_{-}=1.1$, $T_{+}/T_{-}=2.0$, $T_{+}/T_{-}=5.0$, and $T_{+}/T_{-}=10.0$.
The tolerance $r$ used in Alg.~\ref{alg:fluxes} was taken to be equal to $\Delta T/2$; that is, for $\left|T_{+}-T_{-}\right|<r$ the specific heat was simply evaluated at the midpoint $\left(T_{+}+T_{-}\right)/2$, as opposed to using look-up tables for the integral part of the entropy and the internal energy. In practice, significantly smaller values of the tolerance $r$ should be used, as it leads to a noticeable reduction in the error in the flux evaluation, as will be seen further in Sections~\ref{sec:periodic},~\ref{sec:blast}. However, to better showcase how the flux behaviour changes once $\left|T_{+}-T_{-}\right|\geq r$, the relatively large value of $r=\Delta T / 2$ was chosen for this test case.

Figure~\ref{fig:flux_error} shows the error in the evaluation of the flux for the energy $F_{E}^{\mathrm{num},x}$ using Eqn.~\ref{eq:flux-E_x} computed via Algorithm~\ref{alg:fluxes} for various choices of $\Delta T$, as compared to the exact expression for the calorically perfect gas as given by Eqn.~\ref{eq:flux-ec-energy-3}, i.e. the entropy-conservative flux of Chandrashekar~\cite{chandrashekar2013kinetic}. It can be observed that as long as $\left|T_{+}-T_{-}\right|\geq \Delta T / 2$, the error in the flux computed via Algorithm~\ref{alg:fluxes} is extremely small, almost on the order of machine precision. It can also be seen that for a given temperature jump $T_+/T_-$, as soon as $\Delta T$ is chosen large enough that $\left|T_{+}-T_{-}\right|<\Delta T / 2$, the error in the flux becomes significantly larger. Therefore, it can already be seen that using smaller values of the tolerance $r$ will lead to smaller errors, as in fewer cases will the simplified computation be used. For reasonable choices of a temperature discretization step (i.e. 1--10~K), the error still remains  low even when the flux is computed using the simplified approximation with the mean temperature.




\subsubsection{Infinite harmonic oscillator}
Next, we compare the energy flux~(\ref{eq:flux-E_x}) to the flux derived in~\cite{peyvan2023high} for the case of the infinite harmonic oscillator model. The gas was assumed to be molecular oxygen (with a molecular mass $m=5.3134\times10^{-26}$ kg), and a characteristic vibrational temperature of $\theta_v=2273.5$~K was used. The same flow conditions and temperature jumps were considered as for the previous case, and values of the tolerance $r$ was again chosen to be $\Delta T/2$. For the evaluation of both fluxes, the temperatures were assumed to be known exactly, i.e. no impact of the non-linear solver was considered. It should be noted that the numerical flux proposed in~\cite{peyvan2023high} relies upon a Taylor expansion of a hyperbolic sine function of a non-linear combination of $T_-$ and $T_+$, using only the first four terms of the series. For larger temperature jumps, the expansion was found to require one extra term to provide a sufficiently accurate value.

\begin{figure}[h]
    \centering
    \includegraphics[width=0.7\textwidth]{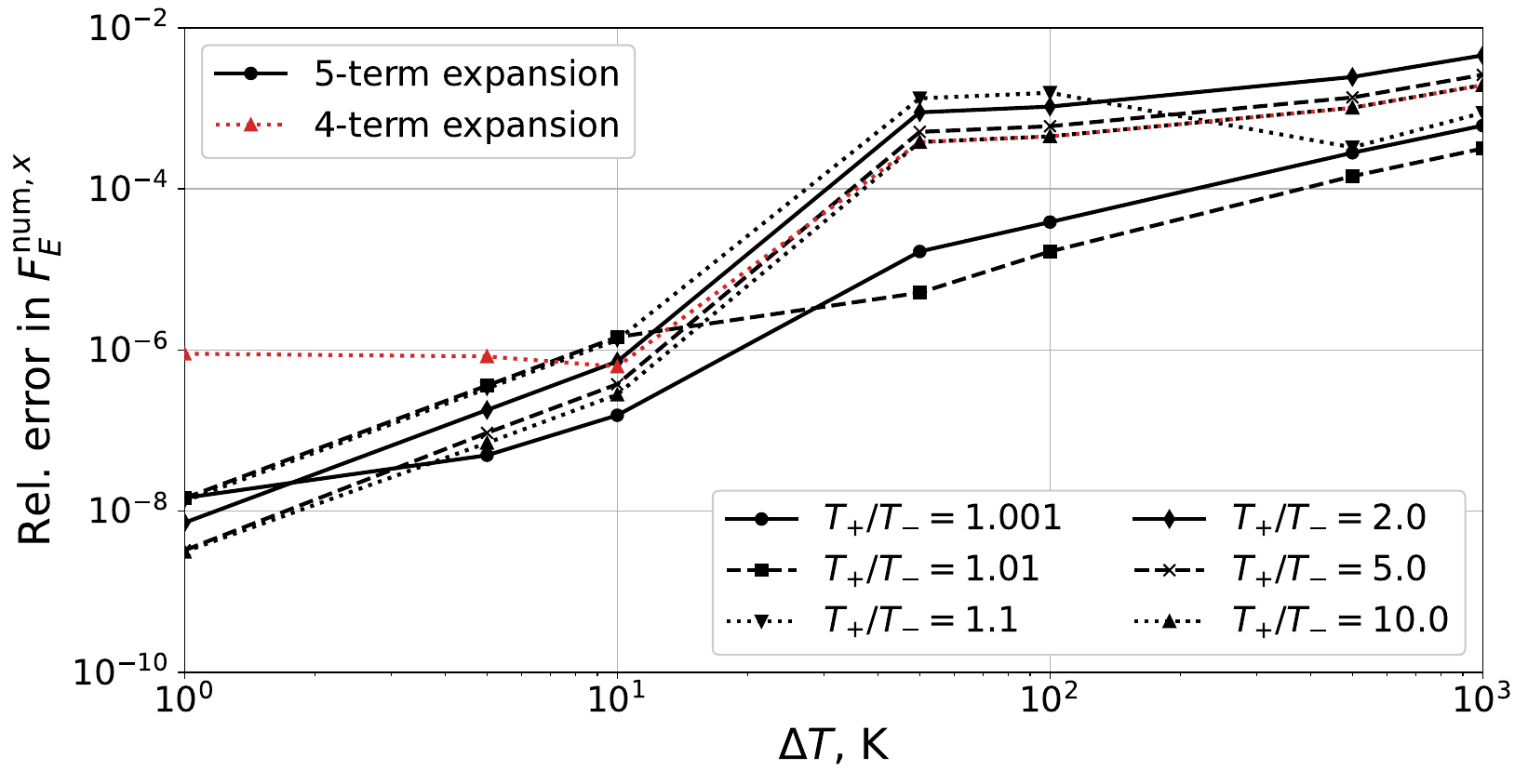}
    \caption{Relative error in the energy flux as a function of the discretization step $\Delta T$, infinite harmonic oscillator.}
    \label{fig:flux_error_iho}
\end{figure}

Figure~\ref{fig:flux_error_iho} shows the relative errors in the energy flux (compared to the flux from~\cite{peyvan2023high} assuming exactly known temperatures) for different temperature jumps as a function of the discretization step $\Delta T$. The black lines show the error relative to the flux of Peyvan et al.~\cite{peyvan2023high} when a 5-term Taylor expansion is used, the red line shows the error relative to the same flux when only 4 terms are used for the larges temperature ratio $T_+/T_-=10$.

Compared to the case of the calorically perfect gas, the choice of temperature discretization step $\Delta T$ has a noticeable impact on the accuracy of the flux, as it affects the accuracy of the computation of $c_v(T)$ and as a result, the accuracy of the computation of $\eta(T)$.

Again, it can be seen that at small values of $\Delta T$, the proposed flux expression~(\ref{eq:flux-E_x}) achieves excellent error. The error in this case is specific to the linear interpolation used in the present work, however, the developed numerical flux formulation allows for use of higher-order interpolations or even exact expressions for energy, specific heats, and the integral part of the entropy. A practical large-scale simulation will also require maintaining a balance between computational cost and accuracy --- as will be discussed in the next subsection, use of simplified interpolations based on tabulated values can provide a noticeable speed-up compared to use of exact analytical expressions. Even at extremely large values of $\Delta T$, i.e. 100--1000~K, the error compared to the exact flux derived in~\cite{peyvan2023high} is still less than 1\%.

It can also be seen that for the largest temperature jump considered, the 4-term expansion used in~\cite{peyvan2023high} proves to be insufficiently accurate. In fact, the number of terms required in the Taylor expansion used can be considered to be an implicit numerical parameter, higher values of which may lead to an increase in accuracy, but also to an increased computational cost. Whereas in the present work the main parameter is the explicitly defined $\Delta T$, which does not affect the computational cost apart from the pre-processing step at the start of the simulation.



\subsubsection{Computational cost comparison}
We carry out a brief comparison between the computational cost of calculating the flux vector~(\ref{eq:flux_fx_rho})--(\ref{eq:flux-E_x}) in conjunction with Eqn.~\ref{eq:eta_compute}, the computational cost of calculating the flux vector as given in~\cite{peyvan2023high}, and the computational cost of evaluating the flux of Chandrashekar~\cite{chandrashekar2013kinetic} for a calorically perfect gas. The flux of Chandrashekar~\cite{chandrashekar2013kinetic} was included in the analysis as a baseline, despite its assumption of a calorically perfect gas.
All fluxes were implemented in the Julia programming language.

\begin{table}[h!]
\centering
\begin{tabular}{ |c|c|c| } 
 \hline

 Present work & Peyvan et al.~\cite{peyvan2023high}, 4-term expansion & Chandrashekar~\cite{chandrashekar2013kinetic} \\ 
 \hline
 17~ns & 23~ns & 11~ns \\
 \hline
\end{tabular}
\caption{Time required for one flux evaluation, omitting computation of temperature from energy.}
\label{tab:cost_comparison}
\end{table}

Table~\ref{tab:cost_comparison} presents the computational costs of one flux evaluation for the three flux functions considered. As expected the flux of Chandrashekar~\cite{chandrashekar2013kinetic} is computationally the most of efficient, since for the calorically perfect gas no interpolations are required. For the flux derived in~~\cite{peyvan2023high}, the original proposed 4-term Taylor expansion of the hyperbolic sine function was used; as discussed above, for large temperature jumps, this may be insufficient and more terms may need to be included, leading to higher computational costs. It can be seen that the proposed flux function is approximately 50\% slower than the flux of Chandrashekar~\cite{chandrashekar2013kinetic}, but still roughly 25\% faster than the analytical flux for the infinite harmonic oscillator model with the 4-term Taylor expansion~\cite{peyvan2023high}.

\begin{table}[h!]
\centering
\begin{tabular}{ |c|c| } 
 \hline

 Inversion using interpolated $\varepsilon_{\mathrm{int}}$, $c_v$ & Inversion using exact formulas for energy and $c_v$\\ 
 \hline
75-85~ns & 400~ns \\
 \hline
\end{tabular}
\caption{Time required to compute $T$ from internal energy.}
\label{tab:cost_comparison_tinv}
\end{table}

As already mentioned, the computational costs of the flux~(\ref{eq:flux_fx_rho})--(\ref{eq:flux-E_x}) and the flux from~\cite{peyvan2023high} do not include the cost of the inversion of the energy in order to obtain the temperature. Assuming that $T_{-}=1000$~K and $T_{+}=1500$~K, and using an initial guess of $300$~K as the starting solution, Mutation++ (compiled using the \texttt{-O3} compiler flag) requires approximately 400~ns to compute  $T_{-}$ and $T_{+}$ from the corresponding energies; it does so using exact analytical formulas for the vibrational energy and specific heat of vibrational degrees of freedom, as given by~(\ref{eq:iho}). Using the piece-wise linear interpolation of the internal energy and specific heat, an implementation of Newton's method in Julia requires approximately 75-85~ns for the computation of $T_{-}$ and $T_{+}$ when also using an initial guess of $300$~K and the same relative and absolute tolerance of $10^{-12}$. For a more complicated internal energy function (such as~(\ref{eq:cutoffosc}) the cost of inverting the  energy to obtain temperature will be even higher. These results are summarized in Table~\ref{tab:cost_comparison_tinv}.

Therefore, it can be concluded that the flux function itself~(\ref{eq:flux_fx_rho})--(\ref{eq:flux-E_x}) is computationally very efficient, especially for cases of a calorically imperfect gas, as it allows for use of arbitrary internal energy functions without the computational complexity of those having an impact on the cost of the flux evaluations. The use of a piece-wise linear approximation for the internal energy and specific heats also leads to a significant speed-up of the Newton solver used to compute the temperature from the internal energy.

\subsection{Periodic flow}\label{sec:periodic}
To assess the performance of the developed flux, we consider a two-dimensional test case with periodic boundary conditions, similar to the one investigated in~\cite{peyvan2023high}, but with a larger temperature variation. A square domain $[0, 1]\times[0, 1]$ was assumed. The gas was taken to be molecular nitrogen, with a molecular mass of $m=4.6517\times10^{-26}$ kg and with the vibrational spectrum modelled by an infinite harmonic oscillator model with $\theta_v$ = 3393.5~K. The temperature discretization $\Delta T$ was taken to be 1~K. The pressure in the domain was taken to be constant as $p=195256$~Pa, and a velocity vector $\mathbf{v}=(11450; 0)^{\mathrm{T}}$ m/s was assumed. A temperature profile varying in the $x$ direction was prescribed as $T(x)=9000.0 + 2000 \sin(2 \pi x)$. A uniform $64\times 64$ grid was used for the discretization of the domain, and 3-rd degree polynomials were used in the DG method. All shock-capturing functions and positivity-preserving limiters were turned off, and the flux~(\ref{eq:flux_fx_rho})--(\ref{eq:flux-E_x}) was also used for the interface fluxes (instead of a Riemann solver), to ensure that no stabilization or dissipation schemes can impact the entropy conservation. The total entropy production rate (over all cells in the simulation domain), which can be defined as
\begin{equation}
    s_{t,\mathrm{tot}}(t) = \sum \frac{\partial s}{\partial \mathbf{u}} \cdot \frac{\partial\mathbf{u}}{\partial t}
\end{equation}
was computed via analysis callbacks available in Trixi.jl. Different values of the tolerance $r$ (as used in Alg.~\ref{alg:fluxes}) were used, ranging from $0.5$~K to $10^{-6}$~K.

\begin{figure}[h]
    \centering
    \includegraphics[width=0.7\textwidth]{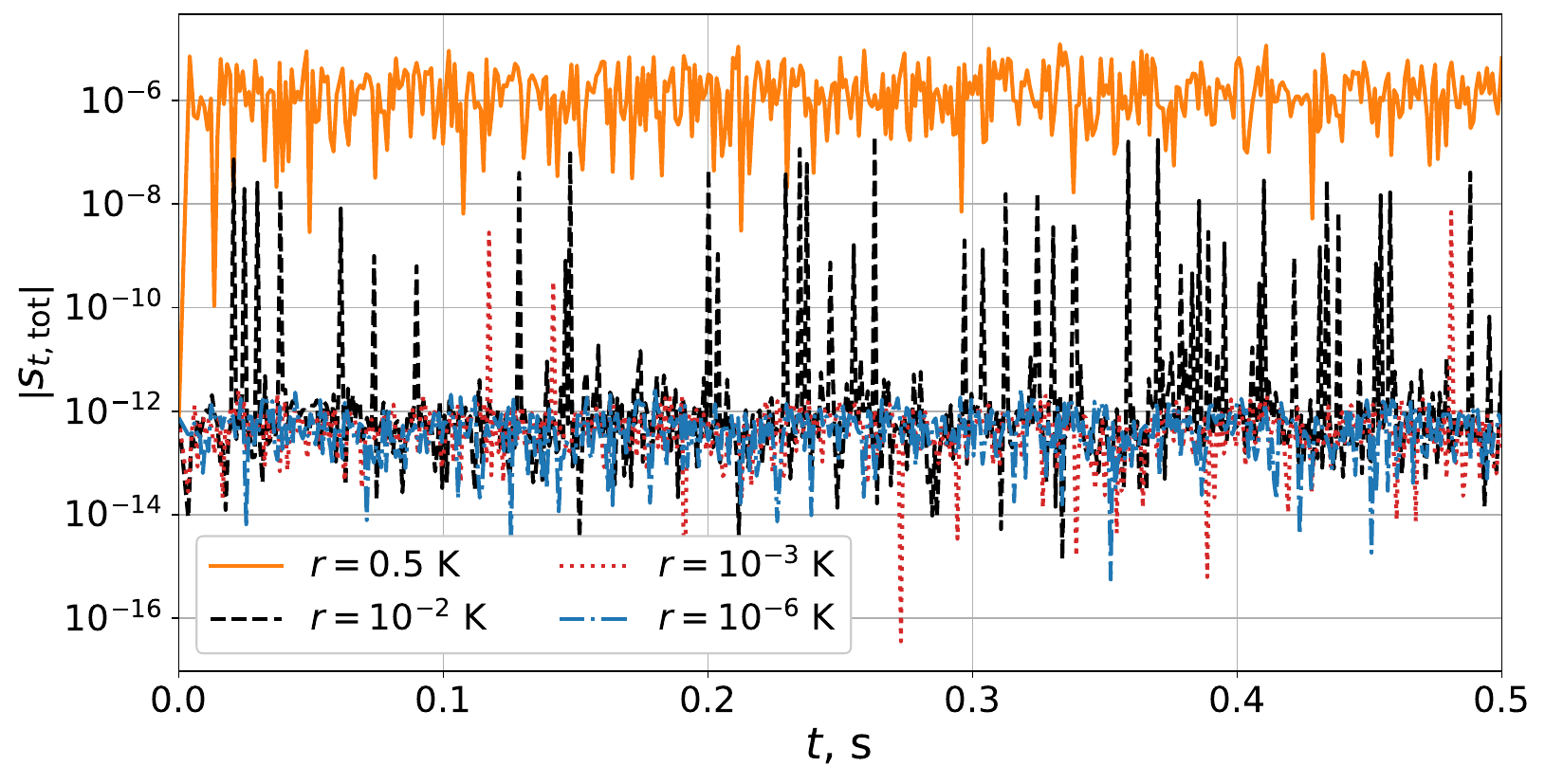}
    \caption{Absolute value of the entropy production rate over the course of the simulation.}
    \label{fig:entropy_production_rate}
\end{figure}
\begin{figure}[h]
    \centering
    \includegraphics[width=0.7\textwidth]{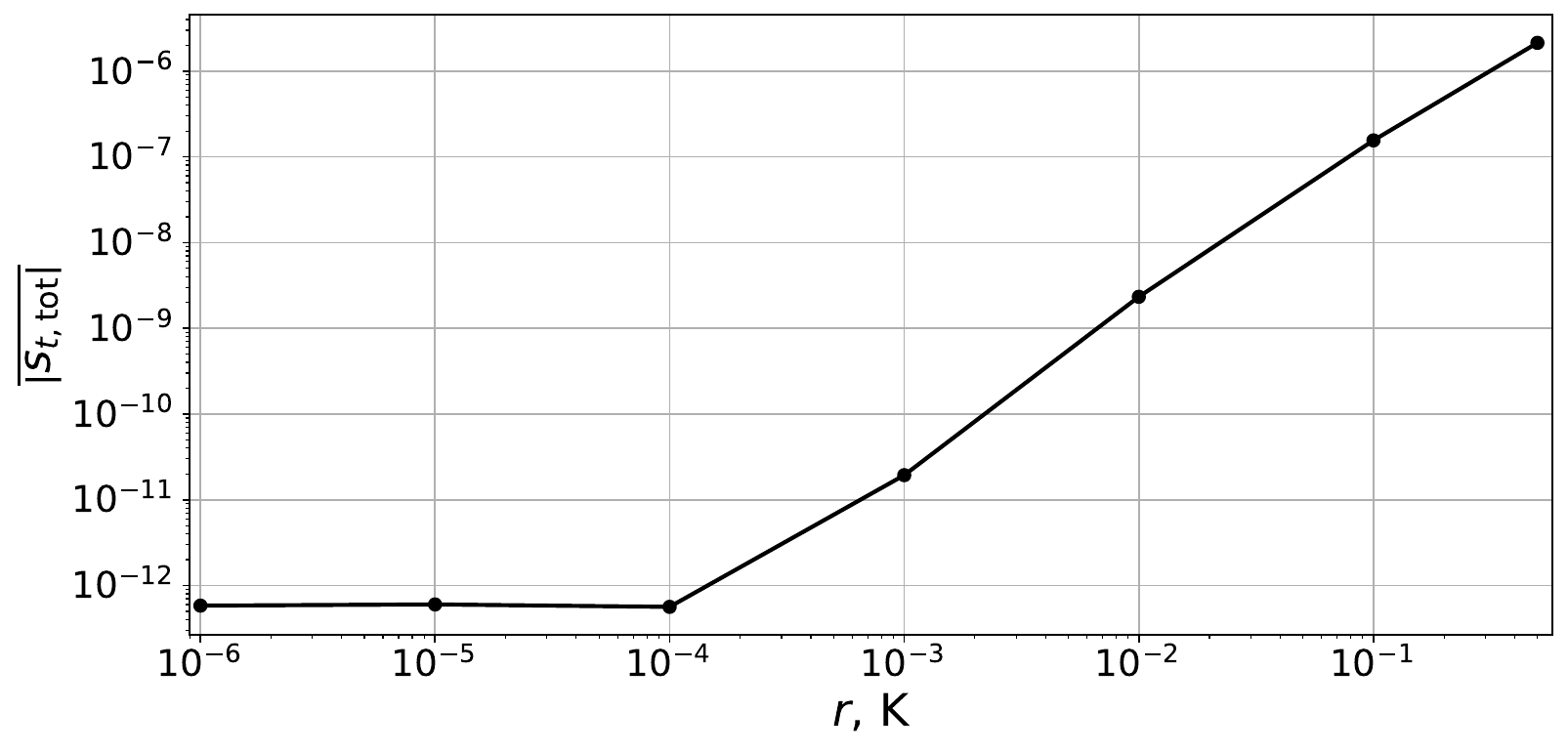}
    \caption{Time-averaged absolute value of the entropy production rate as a function of the parameter $r$.}
    \label{fig:entropy_production_rate_convergence_r}
\end{figure}

Figure~\ref{fig:entropy_production_rate} shows the total entropy production rate plotted over time for different values of the tolerance $r$. It can be clearly seen that the entropy is not conserved to machine precision, due to the errors in the computation of temperature and the integral part of the entropy; however, the overall magnitude of the entropy production rate is still very small. Moreover, lower values of $r$ lead to noticeably lower values of the entropy production rate. To better assess the impact of the value of $r$ on the entropy production rate, we plot the time-averaged absolute value of the entropy production rate as a function of $r$, as shown on Fig.~\ref{fig:entropy_production_rate_convergence_r}. It can be seen that decreasing $r$ up to a value of $10^{-4}$~K leads to decreasing entropy production rates due to errors in the flux function; use of values of $r$ smaller than $10^{-4}$~K however does not lead to lower entropy production rates due to other sources of error becoming dominant.

\subsection{Blast wave}\label{sec:blast}
Next, we consider the weak blast wave condition adapted from~\cite{hennemann2021provably}. The initial conditions are given by
\begin{equation}
    \begin{bmatrix}
           \rho \\
           v_x \\
           v_y \\
           p
         \end{bmatrix} = \begin{bmatrix}
           0.341388 \\
           0.0 \\
           0.0 \\
           101325.0
         \end{bmatrix}\,\mathrm{if}\,\sqrt{x^2+y^2}>0.5;\quad 
    \begin{bmatrix}
           \rho \\
           v_x \\
           v_y \\
           p
         \end{bmatrix} = \begin{bmatrix}
           0.399117 \\
           102.5 \cos(\phi) \\
           102.5 \sin(\phi) \\
           126149.6
         \end{bmatrix}\mathrm{else}.
\end{equation}
Here $\phi=\tan^{-1}(y/x)$. The initial conditions of the problem correspond to a circular-shaped region  with a radius of 0.5~m of higher density and pressure with a constant outward radial velocity.
As in the previous case, the gas was assumed to be molecular nitrogen, with a vibrational spectrum described by the infinite harmonic oscillator model. A square domain of size $[-2, 2]\times[-2, 2]$ was used, discretized by a uniform $64\times 64$ grid. 3-rd degree polynomials were used in the DG method. Similar to the previous case, different values of the tolerance $r$ (as used in Alg.~\ref{alg:fluxes}) were used, ranging from $0.5$~K to $10^{-6}$~K.


\begin{figure}[h]
    \centering
    \includegraphics[width=0.8\textwidth]{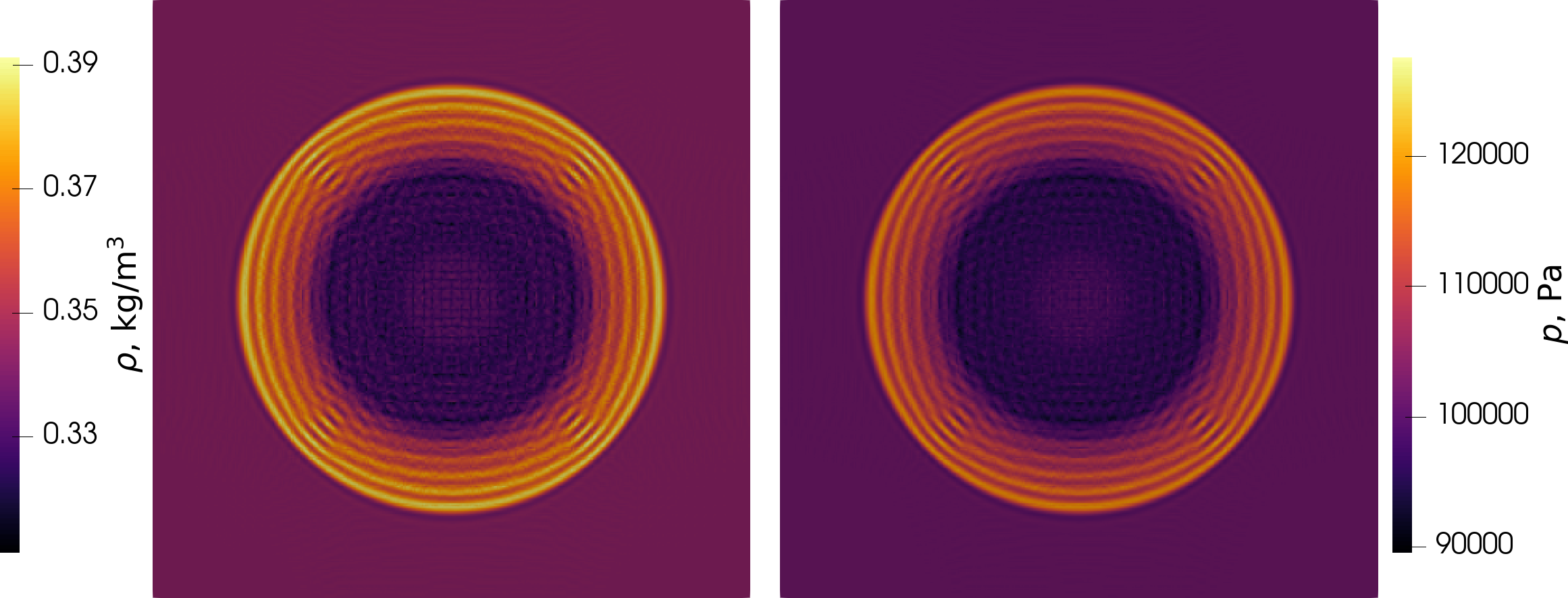}
    \caption{Density (left) and pressure (right) at $t=0.75$~s for the weak blast wave test case, 64$\times$64 grid.}
    \label{fig:p_rho_blastwave}
\end{figure}

Figure~\ref{fig:p_rho_blastwave} shows the density and pressure profiles at $t=0.75$~s. The oscillations in the density and pressure in the radial direction are due to the absence of numerical diffusion in the simulation. The ripples observed in the center of the domain are numerical artifacts are due to the use of a coarse uniform Cartesian grid, which is poorly suited for the radially symmetric problem under consideration. Nevertheless, as expected, the simulation is stable, even without dissipative surface fluxes, shock capturing, and/or positivity preserving limiters.

\begin{figure}[h]
    \centering
    \includegraphics[width=0.825\textwidth]{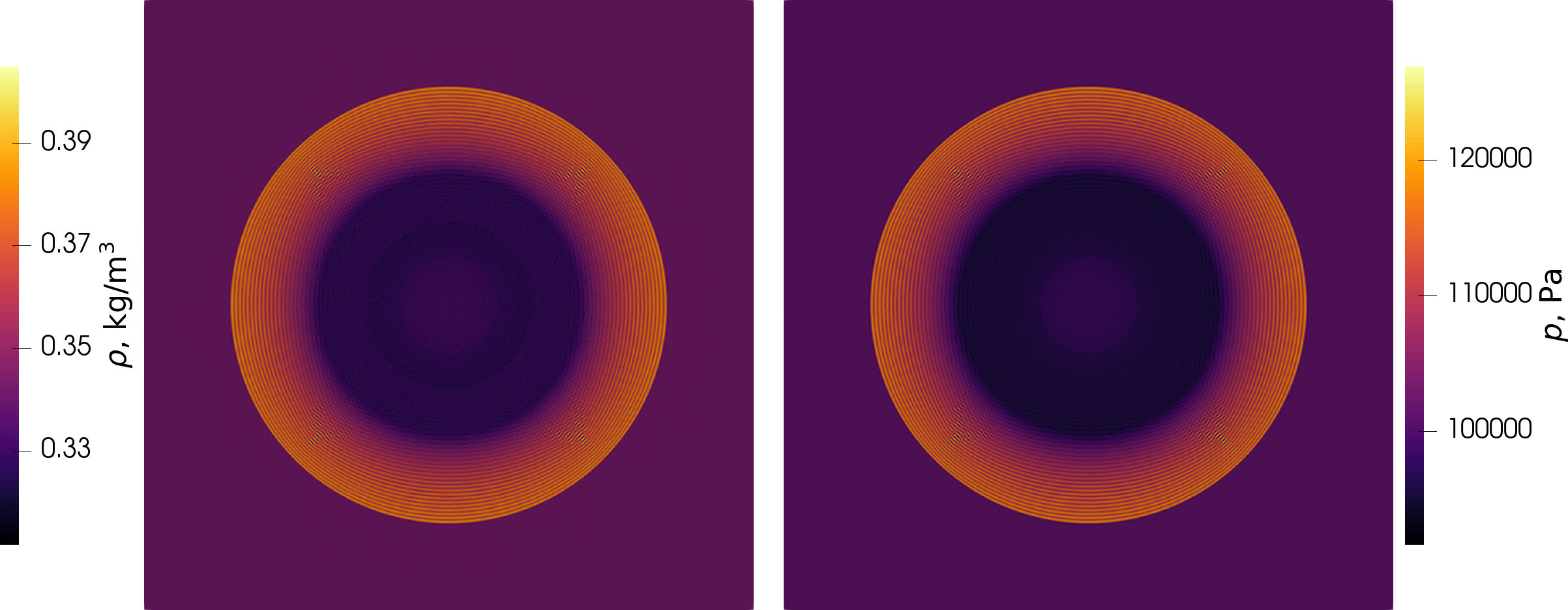}
    \caption{Density (left) and pressure (right) at $t=0.75$~s for the weak blast wave test case, 256$\times$256 grid.}
    \label{fig:p_rho_blastwave_256}
\end{figure}

Figure~\ref{fig:p_rho_blastwave_256} shows the same simulation on a uniform $256\times 256$ grid. The improved grid resolution can be seen to remedy the ripples in the center of the domain, whereas the oscillations across the blast wave front remain, as no numerical diffusion is applied within the solver.

\begin{figure}[h]
    \centering
    \includegraphics[width=0.7\textwidth]{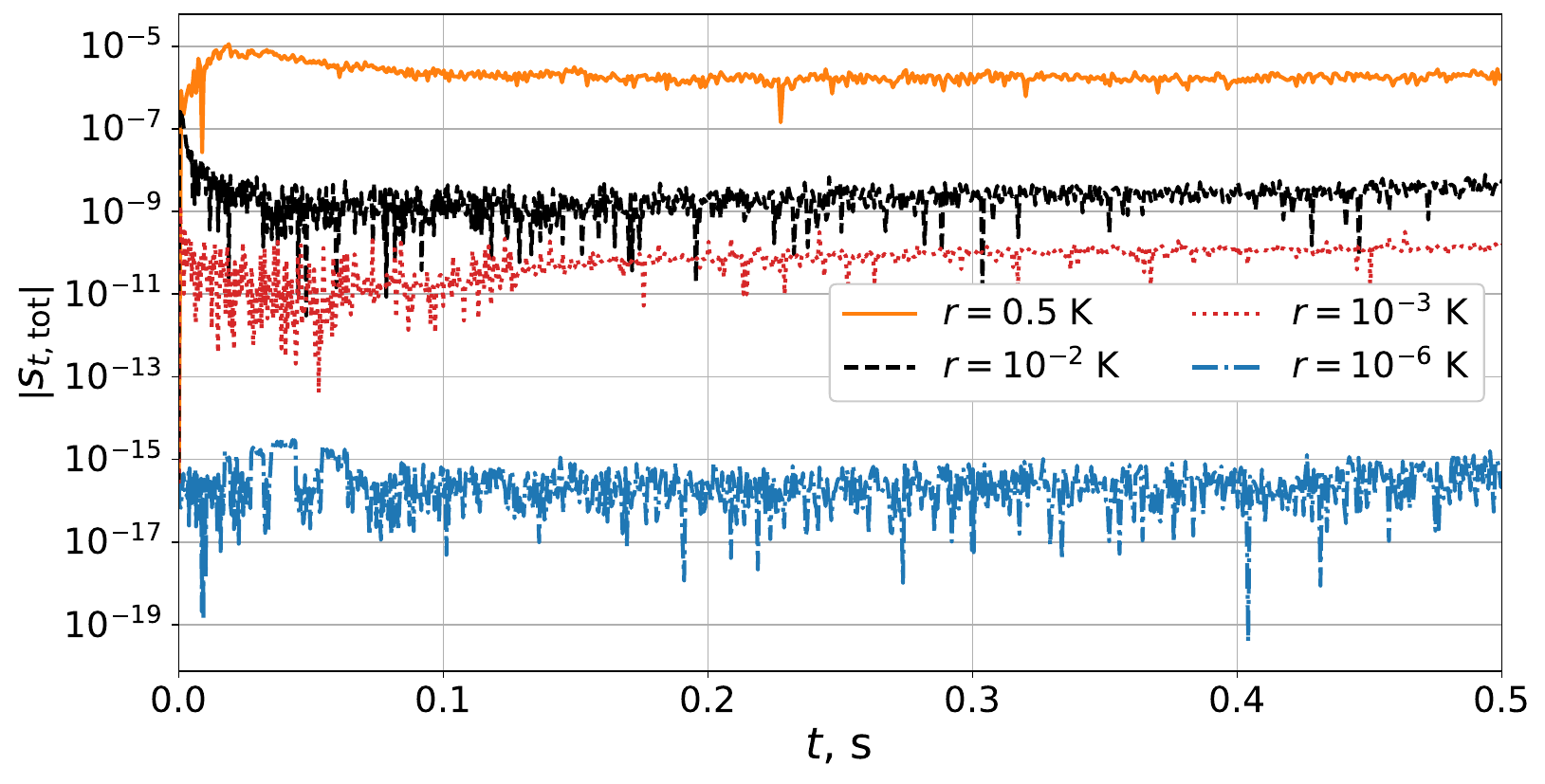}
    \caption{Absolute value of the entropy production rate over the course of the simulation.}
    \label{fig:entropy_production_rate_blastwave}
\end{figure}

\begin{figure}[h]
    \centering
    \includegraphics[width=0.7\textwidth]{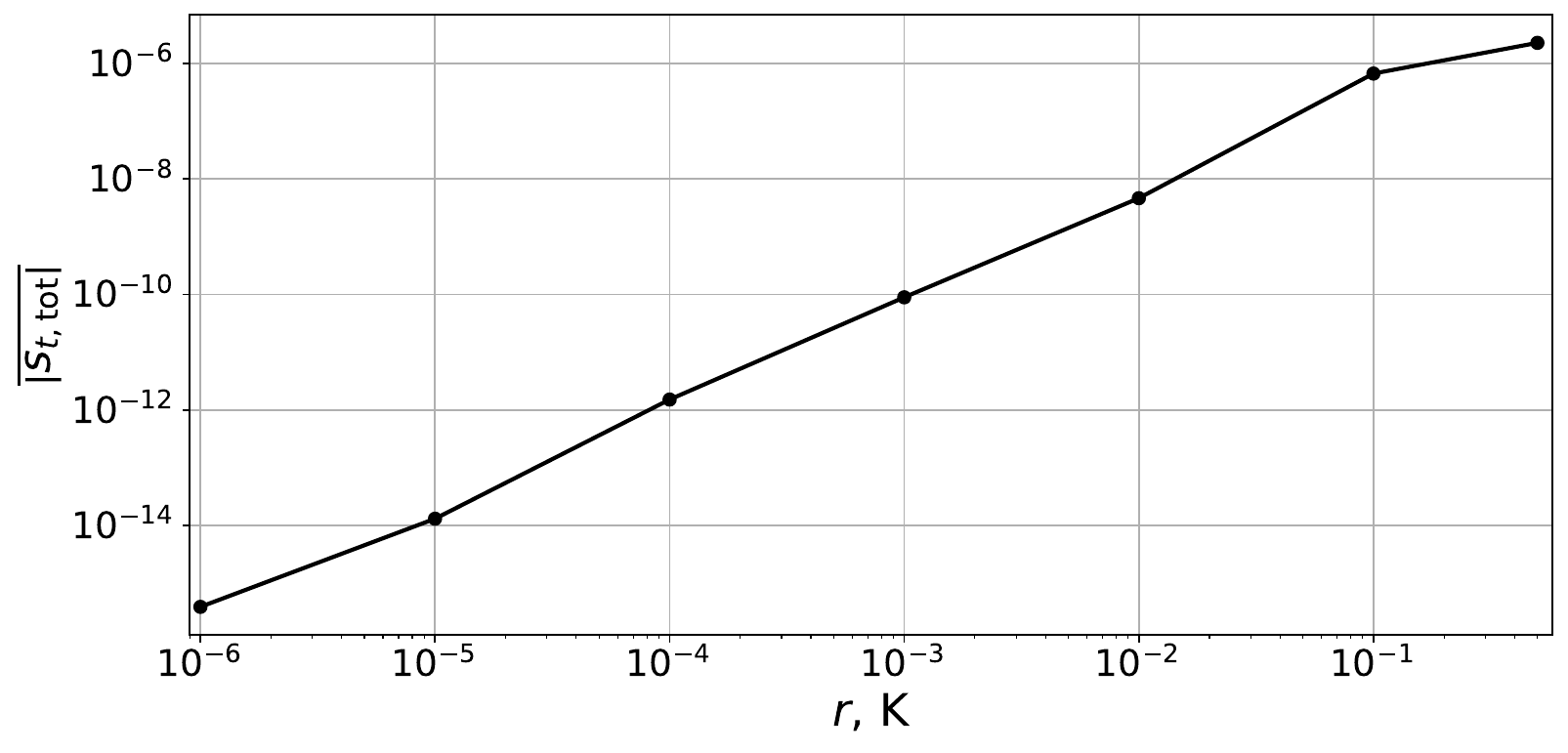}
    \caption{Time-averaged absolute value of the entropy production rate as a function of the parameter $r$.}
    \label{fig:entropy_production_rate_blastwave_convergence_r}
\end{figure}

Figure~\ref{fig:entropy_production_rate_blastwave} shows the total entropy production rate plotted over time for various values of the tolerance $r$ for a temperature discretization step $\Delta T$ of 1~K. Similar to the previous test case, it can be seen that whilst there is an error in the entropy production rate at higher values of $r$, it decreases rapidly when smaller values of $r$ are used. Figure~\ref{fig:entropy_production_rate_blastwave_convergence_r} shows the time-averaged absolute value of the entropy production rate as a function of $r$. The error is almost of the order of machine precision when $r$ is taken equal to 10$^{-6}$~K.

\subsection{Flow around a cylinder}
Finally, we consider a two-dimensional flow around a cylinder as an example of application of the developed flux function to a high-enthalpy flow with strong shocks. The free-stream parameters were chosen to be $v_{\infty}=5956$~m/s, $p_{\infty}=476$~Pa, $T_{\infty}=901$~K, based on the experimental HEG cylinder test-case~\cite{karl2003high}; the cylinder diameter is 0.045~m. The gas was assumed to be either molecular nitrogen with a molecular mass of $m=4.6517\times10^{-26}$ kg or molecular oxygen with a molecular mass $m=5.3134\times10^{-26}$ kg. The flow speed is approximately Mach 10.5, depending on the chosen model for the internal energy. Three models for the internal energy were considered:
\begin{enumerate}
    \item Calorically perfect molecular nitrogen flow with $\gamma=1.4$.
    \item Molecular oxygen flow with the infinite harmonic oscillator model for the vibrational spectrum with $\theta_v=2273.5$~K.
    \item Molecular oxygen flow with the cut-off anharmonic oscillator model for the vibrational spectrum with $\theta_v=2273.5$~K and $\theta_{v,\mathrm{anh}}=17.366$~K, with the cut-off at the dissociation energy of oxygen $D=59364$~K. This corresponds to 36 vibrational levels accounted for.
\end{enumerate}
For the energy and specific heat tables, a step size of $\Delta T=1$~K was used, and the tolerance $r$ in algorithm~\ref{alg:fluxes} was taken to be $10^{-5} \cdot \Delta T$.

\begin{figure}[h]
    \centering
    \includegraphics[width=0.32\textwidth]{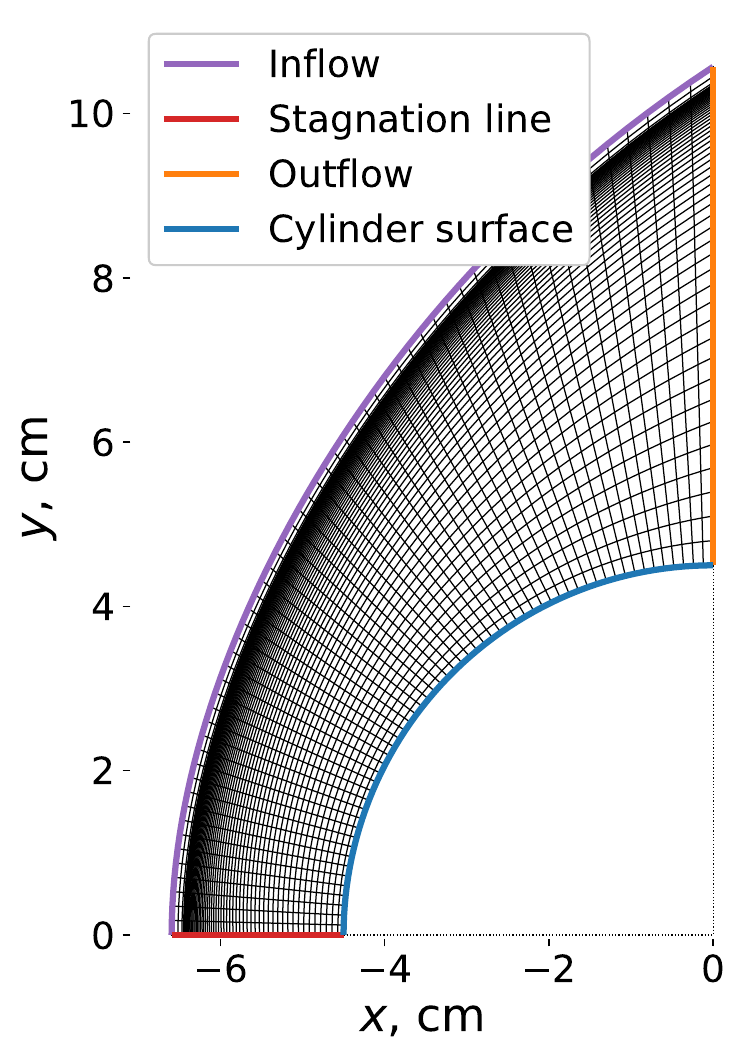}
    \caption{Example of the grid used for the simulations with the infinite harmonic oscillator model.}
    \label{fig:grid}
\end{figure}

The simulations were performed on simple shock-fitted structured $60\times60$ grids. A single cubic curve was fitted to the shock shape. As it does not fit the shock perfectly, some minor errors in the density and pressure are expected near certain regions of the shock, where the mesh and the shock are not well-aligned.
Figure~\ref{fig:grid} shows an example of the simulation grid, specifically, the grid used for simulations with the infinite harmonic oscillator model. The origin of the cylinder is located at $(x,y)=(0,0)$. In the analysis of the simulations presented further below, flow quantities are considered along the stagnation line, given by $y=0$, and the outflow boundary, given by $x=0$. The following boundary conditions were applied in the simulation: supersonic inflow at the inflow boundary (highlighted in purple), supersonic outflow for the outflow boundary (highlighted in orange), and slip conditions at the stagnation line (highlighted in red) and cylinder surface (highlighted in blue).

The developed entropy-conserving flux was used for the volume fluxes, and the dissipative flux~(\ref{eq:flux_diffusive}) was used for the surface fluxes in order to provide sufficient numerical diffusion for the strong shocks present in the flow. Due to the presence of strong shocks in the flow, the sub-cell shock capturing method of Hennemann et al.~\cite{hennemann2021provably} was used. 

The DLR TAU solver~\cite{mack2002validation,hannemann2010closely} was used as a benchmark solver for this case with a similar grid and parameters governing the gas properties. The AUSM+ flux~\cite{liou1996sequel} was used, along with local time-stepping and a first-order implicit Backward Euler solver. 

Firstly, we consider the simplest case, that of the calorically perfect molecular nitrogen flow with $\gamma = 1.4$. This case has been chosen to verify the basic solver components in the Trixi.jl implementation, such as handling of curvilinear structured meshes, boundary conditions, and shock capturing.

\begin{figure}[h]
    \centering
    \includegraphics[width=0.25\textwidth]{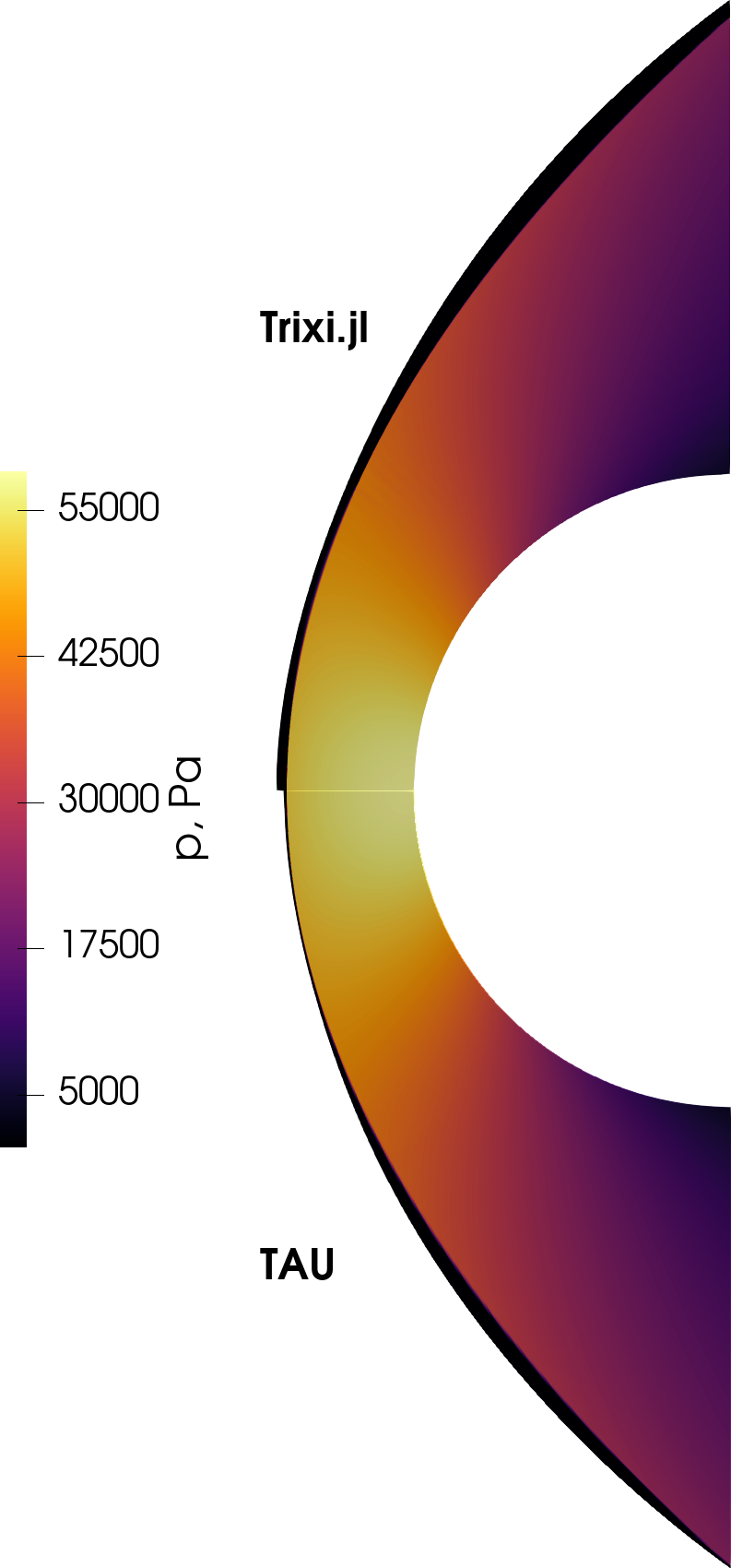}
    \caption{Flow-field pressure computed the developed entropy-conservative approach implemented in Trixi.jl (top) and using the TAU code (bottom) for the case of a constant specific heat ($\gamma=1.4$).}
    \label{fig:60x60_pressure_g14}
\end{figure}
Figure~\ref{fig:60x60_pressure_g14} shows the computed flow-field pressure for TAU and Trixi.jl simulations grid, with the Trixi.jl results computed using 2-nd order polynomials. Excellent agreement  in the shock position and overall pressure profiles can be observed between the solvers. The Trixi.jl solution used a grid with several larger cells in the pre-shock region, therefore inflow region is somewhat larger on the top subplot. For a more quantitative comparison, we evaluate the flow properties along the stagnation line.

\begin{figure}[h]
    \centering
    \includegraphics[width=0.45\textwidth]{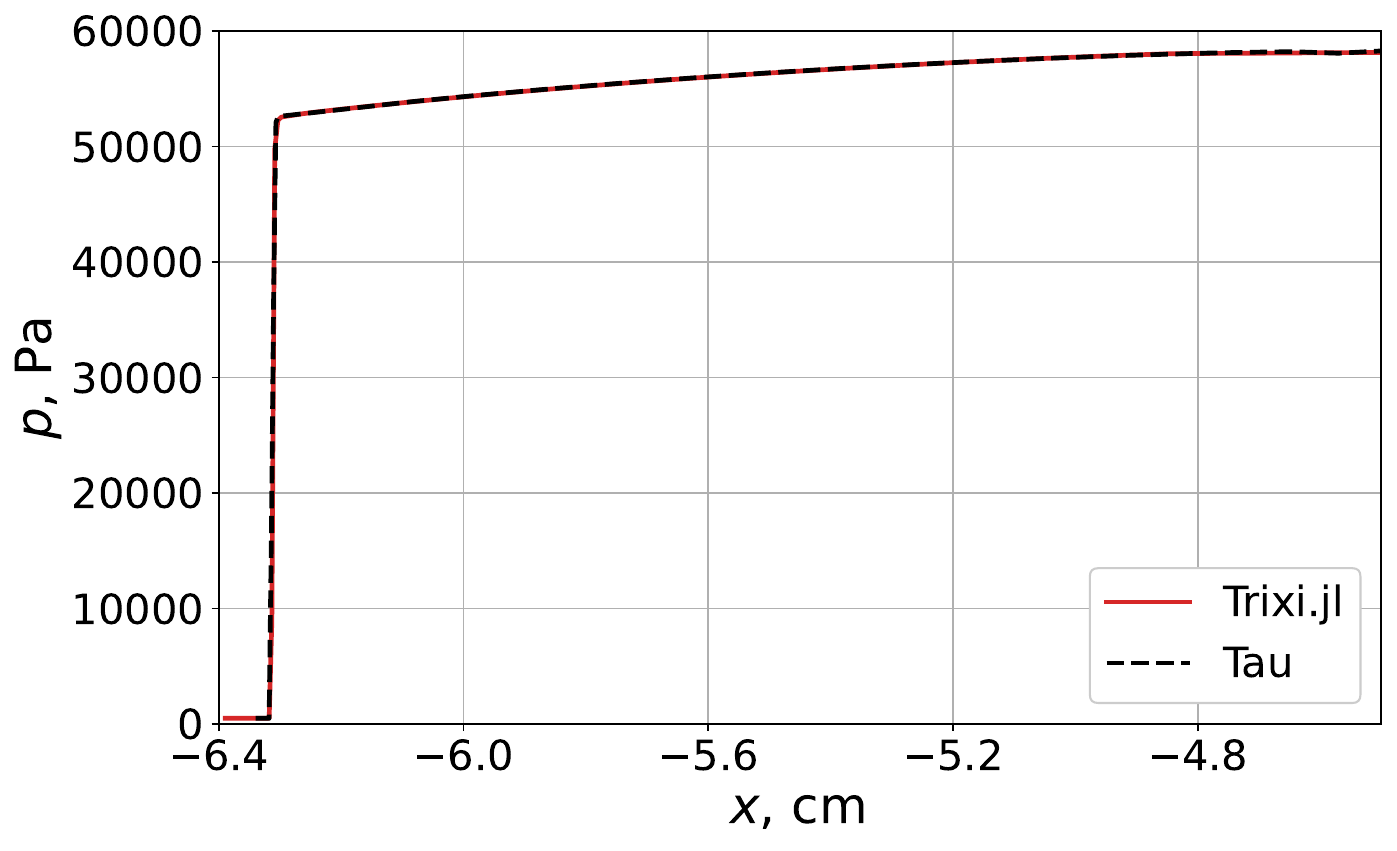}
    \includegraphics[width=0.45\textwidth]{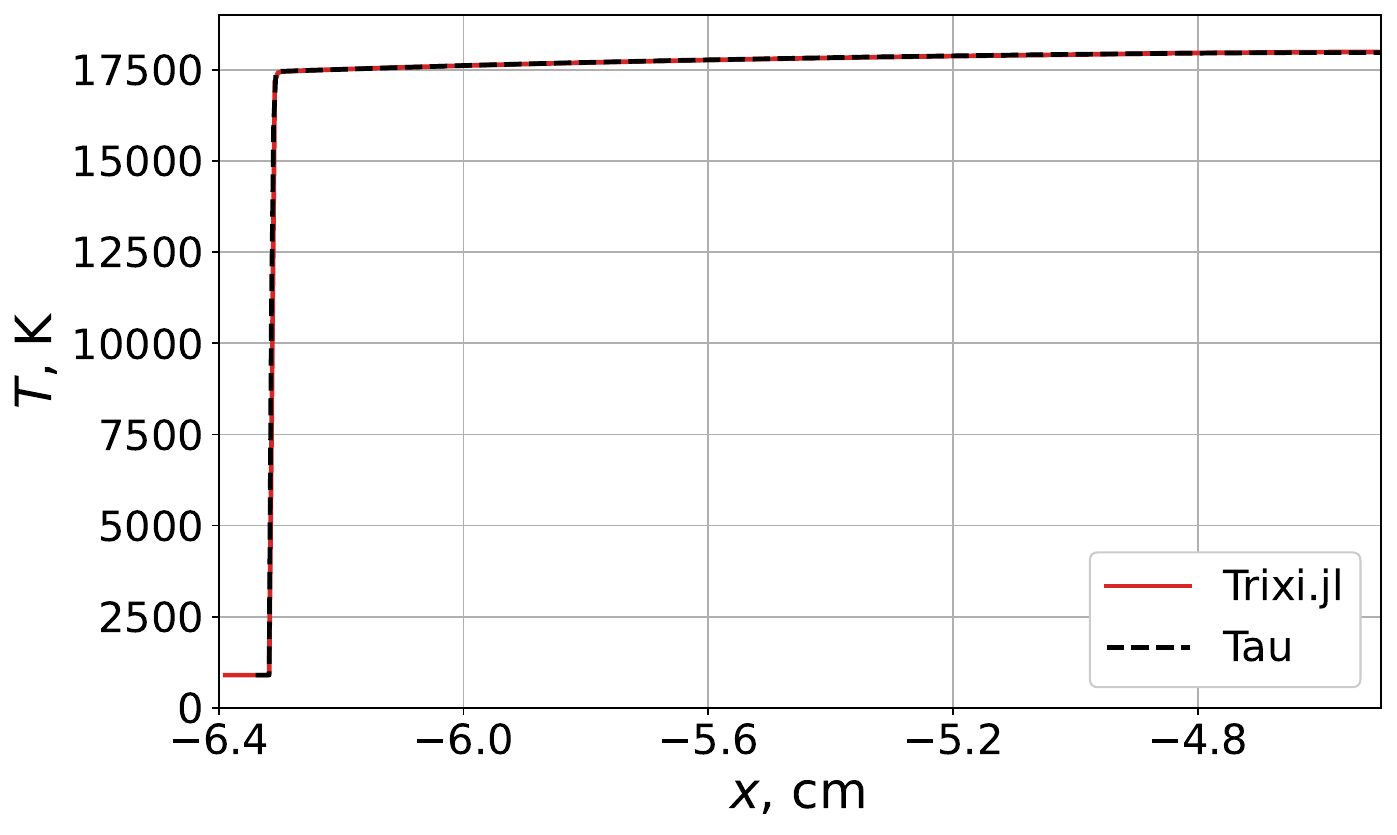}
    \caption{Pressure (left) and temperature (right) along the stagnation line for the case of $\gamma=1.4$.}
    \label{fig:60x60_stagn_g14}
\end{figure}

Figure~\ref{fig:60x60_stagn_g14} shows the pressure and temperature along the stagnation line (the surface of the cylinder is located at $x=-4.5$~cm) computed using TAU and Trixi.jl, with the Trixi.jl simulation using 2-nd order polynomials. Excellent agreement can be seen between the solvers, thus confirming the validity of the basic solver components.

\begin{figure}[h]
    \centering
    \includegraphics[width=0.25\textwidth]{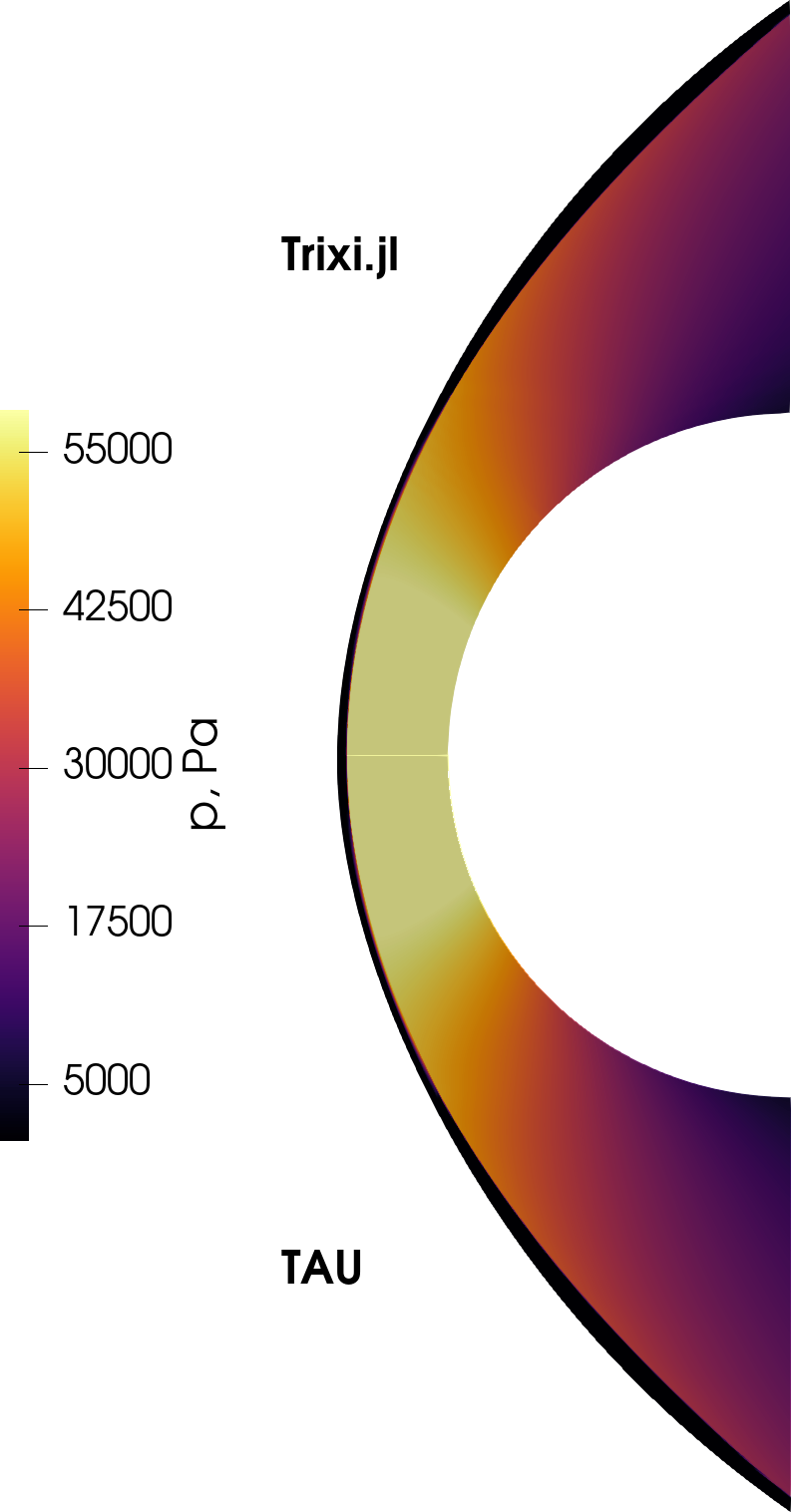}
    \caption{Flow-field pressure computed the developed entropy-conservative approach implemented in Trixi.jl (top) and using the TAU code (bottom) for the case of the infinite harmonic oscillator.}
    \label{fig:60x60_pressure_iho}
\end{figure}

Now we consider the cases with temperature-dependent specific heats.  Figure~\ref{fig:60x60_pressure_iho} shows the computed flow-field pressure for the case of the vibrational spectrum modelled by an infinite harmonic oscillator. Similar to the previous case, excellent qualitative agreement can be seen between the simulations in the whole domain.

\begin{figure}[h]
    \centering
    \includegraphics[width=0.45\textwidth]{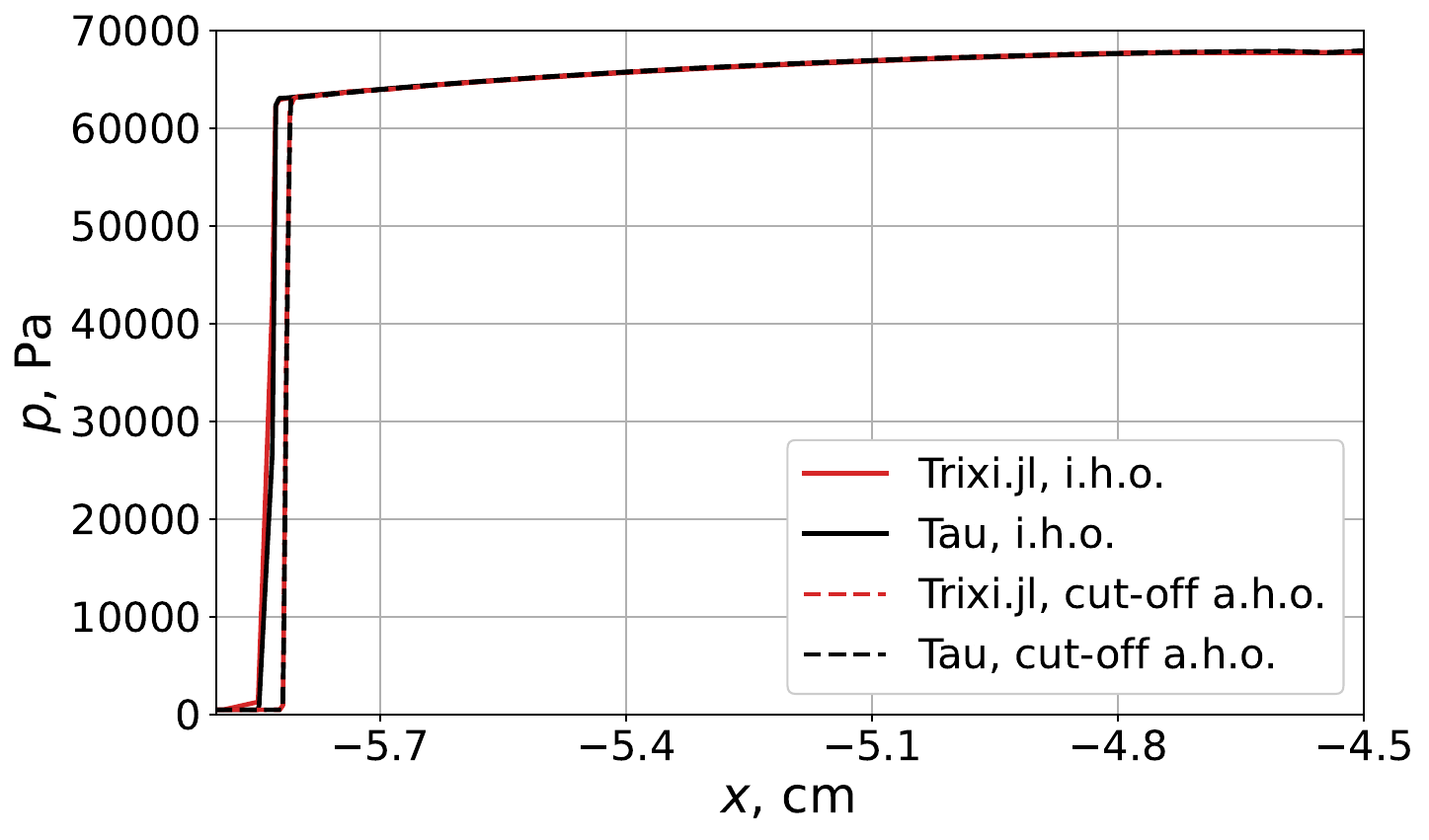}
    \includegraphics[width=0.455\textwidth]{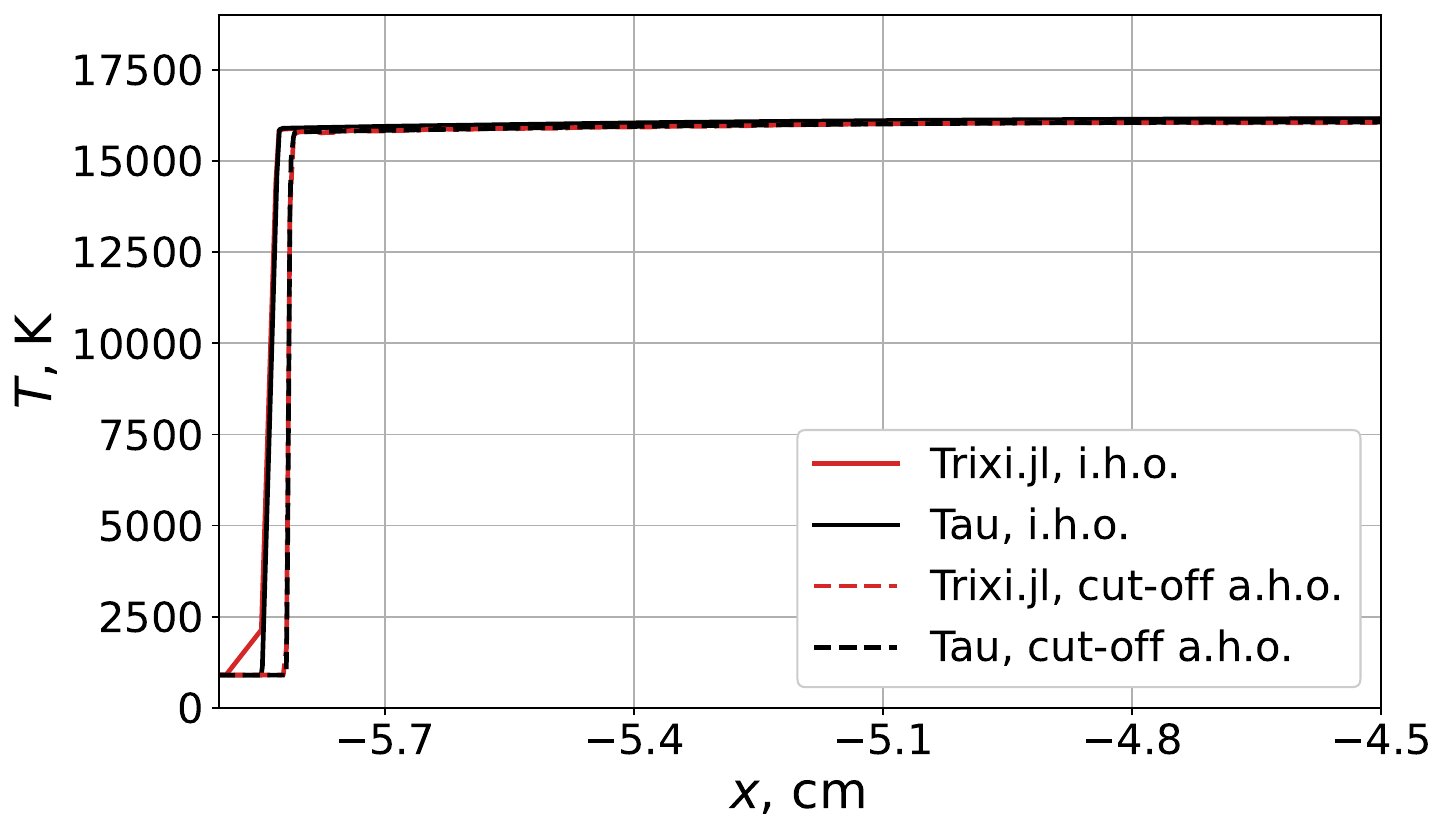}
    \caption{Pressure (left) and temperature (right) along the stagnation line for two different vibrational spectra models.}
    \label{fig:60x60_stagn_iho}
\end{figure}

For a more detailed analysis of this case and that of the cut-off anharmonic oscillator, we look at the flow-field properties along the stagnation line $y=0$ (assuming that the center of the cylinder is located at $(0,0)$). Figure~\ref{fig:60x60_stagn_iho} shows the pressure and temperature along the stagnation line (the surface of the cylinder is located at $x=-4.5$~cm) computed using TAU and Trixi.jl for the two different vibrational spectrum models, i.e. infinite harmonic oscillator and the cut-off anharmonic oscillator. Compared to the case of the calorically perfect gas with $\gamma=1.4$, both models for the vibrational spectrum have lower values of $\gamma$ (as seen on Fig.~\ref{fig:gamma}), thus leading to smaller shock stand-off distances~\cite{sinclair2017theoretical}. The use of the cut-off harmonic oscillator model leads to smaller values of $\gamma$ compared to the infinite harmonic oscillator model, and therefore the stand-off distance is smaller for this case. 

Finally, we analyse the flow quantities in the expansion region, namely, along the outflow boundary given by $x=0$.
Slightly more differences were observed between Trixi.jl simulations using different grid and polynomial orders in this region, therefore, more results are presented: additional simulations were carried out for 2-nd, 3-rd, and 4-th order polynomials on $30\times30$ and $60\times60$ grids.

\begin{figure}[h]
    \centering
    \includegraphics[width=0.45\textwidth]{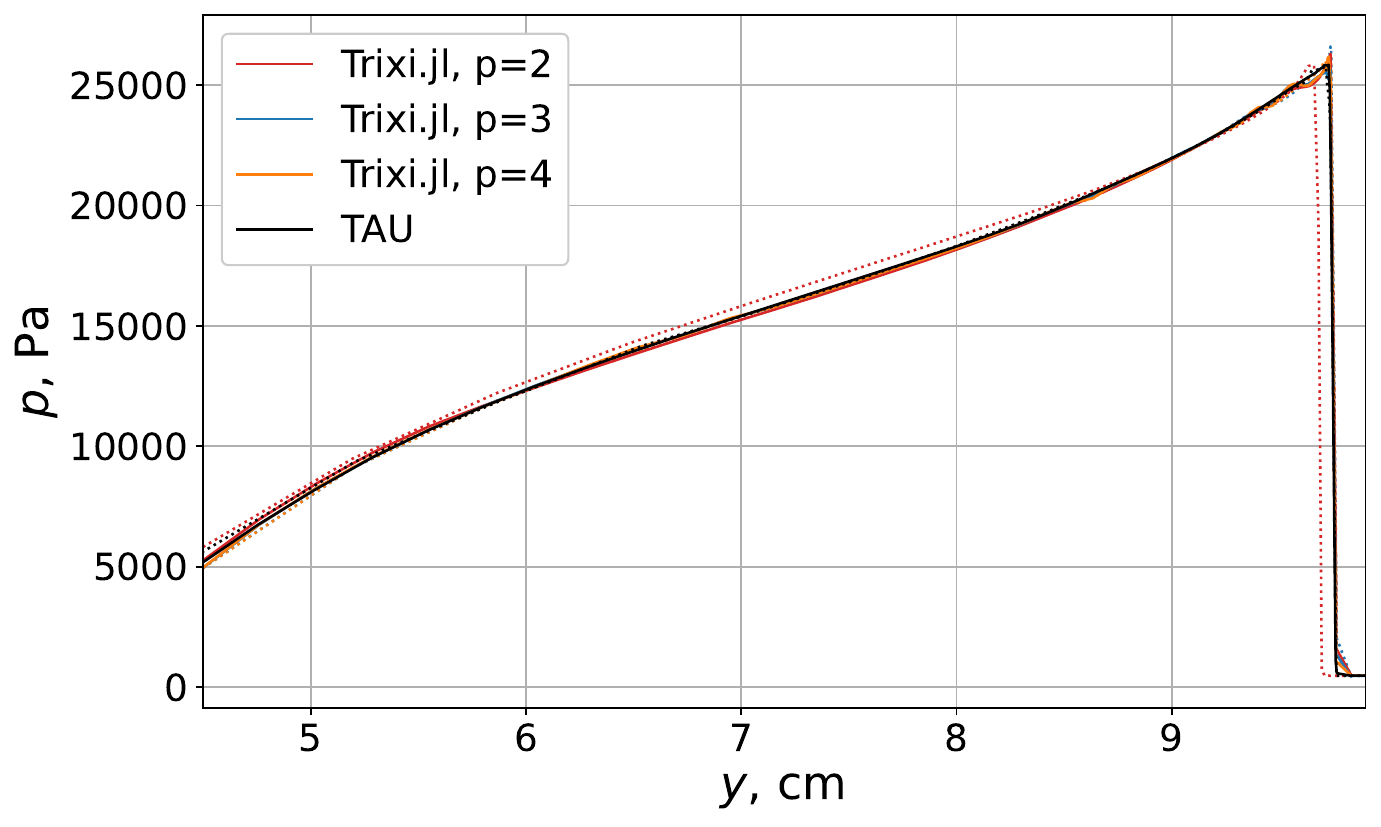}
    \includegraphics[width=0.45\textwidth]{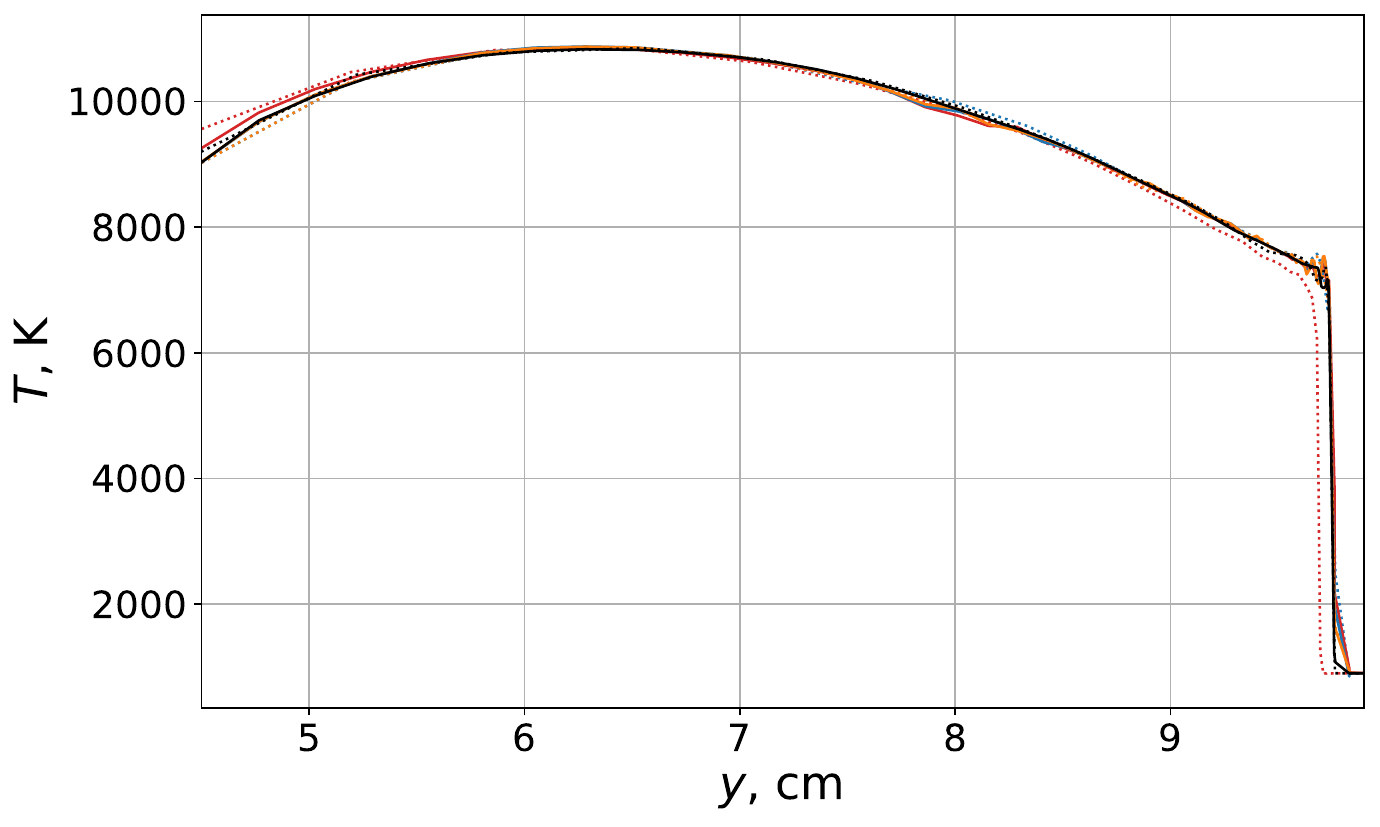}
    \caption{Pressure (left) and temperature (right) along the outflow boundary, infinite harmonic oscillator model. Solid lines correspond denote simulations conducted on a $60\times60$ grid, dotted lines denote simulations conducted on  a $30\times30$ grid.}
    \label{fig:60x60_should_iho}
\end{figure}

Figure~\ref{fig:60x60_should_iho} shows the pressure and temperature in the expansion region along the line $x=0$ (the surface of the cylinder is located at $y=4.5$~cm) for the infinite harmonic oscillator model of the vibrational spectrum. Solid lines correspond to the $60\times60$ grid, dotted lines correspond to the $30\times30$ grid.  We see that all simulation results agree very well, with only the 2-nd order DG simulation on the coarse $30\times30$ grid (dotted red line) deviating slightly in terms of pressure and temperature from the other results. The Trixi.jl-based DG simulations predict a sharper shock, but use of the higher-order polynomials also leads to the presence of irregularities (seen on the temperature profile at $y=8$~cm), as well as slightly stronger oscillations near the shock (seen on the temperature profile near $y=9.7$~cm).

\begin{figure}[h]
    \centering
    \includegraphics[width=0.45\textwidth]{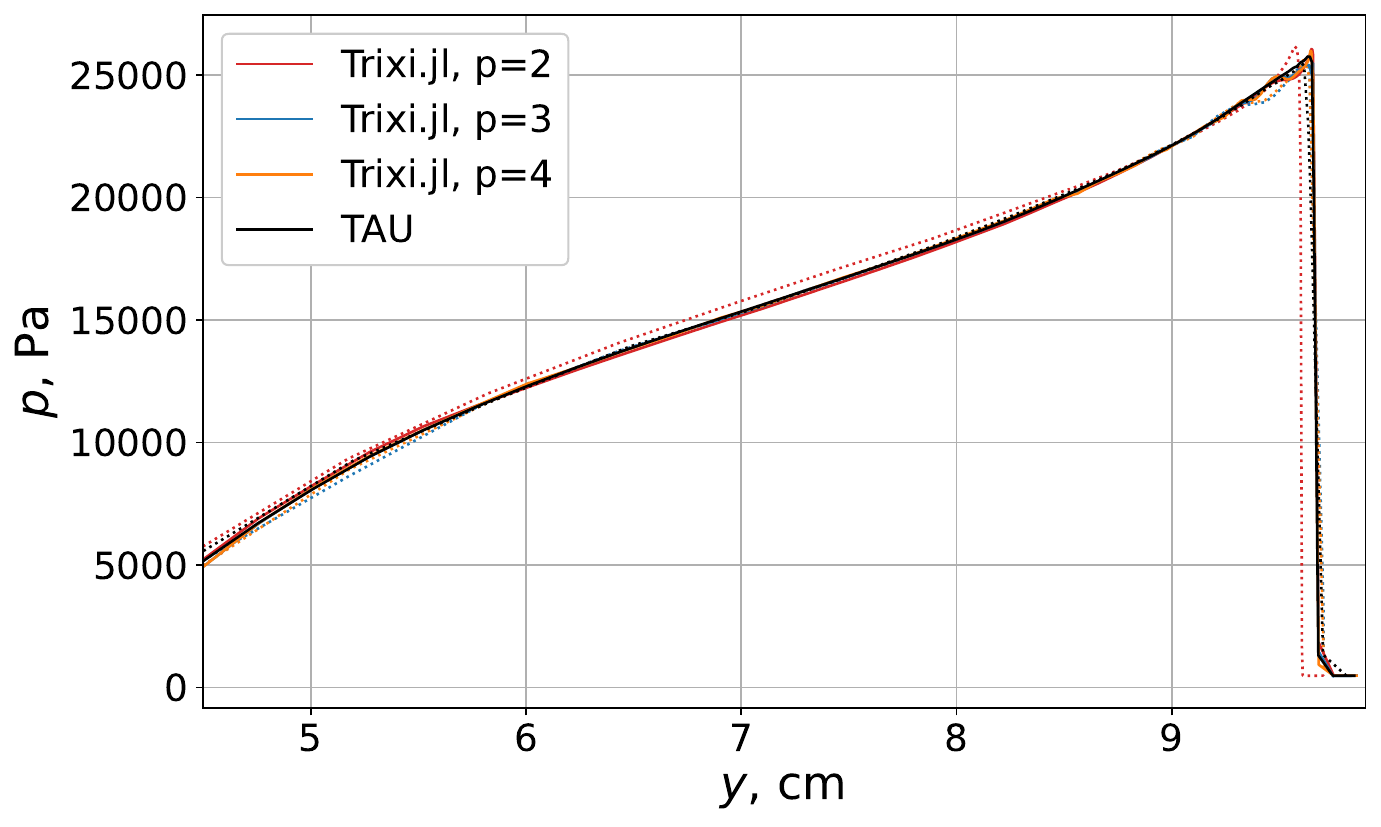}
    \includegraphics[width=0.45\textwidth]{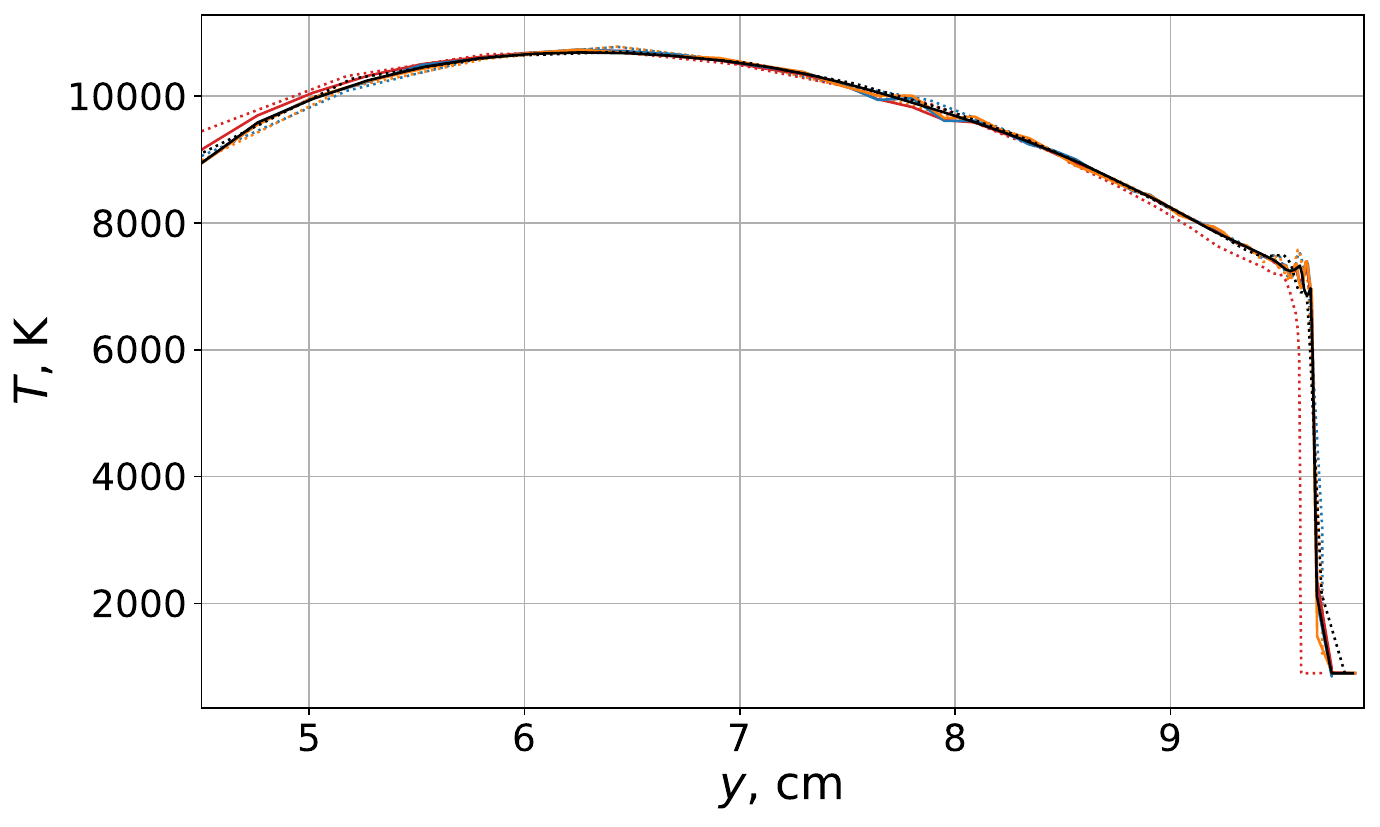}
    \caption{Pressure (left) and temperature (right) along the outflow boundary, cut-off anharmonic oscillator model. Solid lines correspond denote simulations conducted on a $60\times60$ grid, dotted lines denote simulations conducted on  a $30\times30$ grid.}
    \label{fig:60x60_should_anho}
\end{figure}

Finally, Fig.~\ref{fig:60x60_should_anho} shows the shows the pressure and temperature along the line $x=0$ for the cut-off anharmonic oscillator model. Similarly to the infinite harmonic oscillator case, only the DG simulation using 2-nd order polynomials on a $30\times30$ grid displays any noticeable discrepancy in terms of the pressure from the other results. The higher-order simulations also exhibit slight oscillations in the temperature; however, they were not found to affect the stability of the simulations in general.

Therefore, it can be concluded that the developed approach for entropy-stable simulations of high-enthalpy flows of gases with internal degrees of freedom provides accurate results, and has good properties in terms of computational performance.

\section{Conclusion}\label{sec13}
An approach to the computation of entropy-conserving fluxes for gases with arbitrary internal energies has been developed, based on the piece-wise linear approximation of the temperature-dependent internal energy and specific heat. For a calorically perfect gas, the fluxes have been shown to reduce to well-known expressions available in literature. A significant advantage of the proposed approach is its independence on the exact expressions for the internal energy, as they are only required in the pre-computation step, and the flow simulations only perform linear interpolation, the cost of which is independent of the exact expressions for the internal energy and specific heat capacities.

The flux function has been implemented in the Trixi.jl framework and simulations of a Mach 10.5 flow over a cylinder were performed for three different internal energy functions. Comparisons with computations carried out using the DLR TAU solver show excellent agreement in the flow-field quantities.

Further future extensions of the work include the generalization of the scheme to multi-species reacting flows, and its application to viscous flows.

\section*{Acknowledgments}
This work has been supported by the German Research Foundation within the research unit DFG-FOR5409. Georgii Oblapenko also acknowledges the support of the Humboldt Foundation for his guest research stay at the German Aerospace Center (DLR) Institute of Aerodynamics and Flow Technology in G\"{o}ttingen, during which the results with the DLR TAU solver had been obtained. The authors also thank Dr. Moritz Ertl of DLR Institute of Aerodynamics and Flow Technology G\"{o}ttingen for his help with the DLR TAU simulations.



\end{document}